\documentclass[12pt]{amsart}

\usepackage{mathabx, macros, amsaddr}

\usepackage[margin=1in,bottom=1.2in,footskip=0.5in]{geometry}
\spacing{1.2}

\allowdisplaybreaks

\def\del{\partial}

\def\rC{\mathrm{C}}

\def\fgl{{\mathfrak{gl}}}
\DeclareMathOperator{\Proj}{Proj}
\DeclareMathOperator{\aut}{aut}
\def\A{\ensuremath{\mathbb{A}}}
\def\cat#1{\ensuremath{\text{\textsf{#1}}}}
\DeclareMathOperator{\Mod}{\cat{Mod}}
\DeclareMathOperator{\Dol}{Dol}
\def\PV{{\rm PV}}

\begin{document}

\title[Twisting pure spinor superfields]{Twisting pure spinor superfields, \\ with applications to supergravity}

\author{Ingmar Saberi}
\address{Ludwig-Maximilians-Universit\"at M\"unchen \\ Fakult\"at f\"ur Physik \\ Theresienstra\ss{}e 37 \\ 80333 M\"unchen \\ Deutschland}
\email{i.saberi@physik.uni-muenchen.de}

\author{Brian R. Williams}
\address{School of Mathematics\\
University of Edinburgh \\ 
Edinburgh \\ 
UK}
\email{brian.williams@ed.ac.uk}

%\date{Began: December 7, 2020. Last compiled: \today.}
\date{June 28, 2021}

\begin{abstract}
We study twists of supergravity theories and supersymmetric field theories, using a version of the pure spinor superfield formalism. Our results show that, just as the component fields of supersymmetric multiplets are the vector bundles associated to the equivariant Koszul homology of the variety of square-zero elements in the supersymmetry algebra, the component fields of the holomorphic twists of the corresponding multiplets are the holomorphic vector bundles associated to the equivariant Koszul homology of square-zero elements in the twisted supersymmetry algebra. The BRST or BV differentials of the free multiplet are induced by the brackets of the corresponding super Lie algebra in each case. We make this precise in a variety of examples; applications include rigorous computations of the minimal twists of eleven-dimensional and type IIB supergravity, in the free perturbative limit. The latter result proves a conjecture by Costello and Li, relating the IIB multiplet directly to a presymplectic BV version of minimal BCOV theory.
\end{abstract}

\maketitle

\newpage

\tableofcontents
\smallskip
\centerline{\em To Martin Cederwall}
\smallskip

\section{Introduction}

Twists of supersymmetric gauge theories have been a rich source of new ideas and results in  both physics and mathematics over the last few decades. The influence of topological quantum field theory in low-dimensional topology can hardly be overstated; twisting can be used to extract such topological field theories from supersymmetric theories of physical interest, as in Witten's approach  to Donaldson theory for four-manifolds. In turn, expectations about the original supersymmetric theory (or different approaches to understanding its behavior, for example via dualities) can produce powerful and surprising new insights about the invariants captured by its twists; Seiberg--Witten theory is a noted example of this kind.

From the mathematician's perspective, there are many questions and goals that arise immediately from this line of thinking. For a start, one could hope to understand all possible twists of a large class of supersymmetric theories, so as to explore as many different mathematical applications as possible. One might also ask to understand all of the different dualities that have been proposed in the physics literature, and look for mathematical incarnations or manifestations of each of them. And one might hope to formalize the construction of the (untwisted) supersymmetric theories themselves, so as to better understand the spectrum of possibilities there, as well as the logic that relates various twists and duality frames to one another. 

Among dualities of physical interest, the AdS/CFT correspondence remains relatively open territory in terms of potential mathematical applications. This is likely due, at least in part, to the fact that it involves theories of gravity, which are harder to deal with than standard field theories; even the notion of twisting requires care in this setting. There has, however, been notable progress in this direction over the last few years. Costello and Li proposed a natural definition of twisting for supergravity theories, and gave a conjectural description of the holomorphic twist of type IIB supergravity theory, motivated by ideas from topological string theory. Using their description, they showed relations between the large-$N$ limit of (holomorphically twisted) brane worldvolume theories and their holomorphic closed string field theory. Further work on a mathematical approach to holography via Koszul duality can be found in~\cite{CostelloPaquette, CostelloMtheory1}.  

Costello and Li did not derive their description of holomorphically twisted supergravity directly from a target-space description of the supergravity multiplet; one imagines that this was perhaps due to the difficulty of giving a complete formulation of the homotopy module structure in a component-field BV description. One central result of this paper is a proof of Costello and Li's conjecture: we directly compute the holomorphic twists of both type IIB and eleven-dimensional supergravity, at the level of free (degenerate) BV theories. 

The key technical tool in our computation is a version of the pure spinor superfield formalism. These techniques have been developed for some time in the physics literature, notably in Berkovits' work on the pure spinor superstring and in the work of Cederwall and collaborators~\cite{BerkovitsCovariant, Cederwall}. In particular, Cederwall gave a formulation of perturbative eleven-dimensional supergravity using a pure spinor superfield~\cite{Ced-towards,Ced-11d}, building on earlier work by Howe and Berkovits~\cite{Howe11d,Berkovits11d}. The formalism was reinterpreted in~\cite{NV}, where a description of the type IIB supergravity multiplet was presented\footnote{The construction of this multiplet was also known to Martin Cederwall in unpublished work.} and the close connection of the procedure to twisting was also observed. The relevant spaces are not always spaces of  pure spinors in the sense of Cartan; rather, they are the spaces $Y$ of square-zero elements in (the odd part of) the supersymmetry algebra. As such, the same object appears in two seemingly different roles: on the one hand, the twists of any supersymmetric theory fit together into a natural family over~$Y$; on the other hand, an \emph{untwisted} supersymmetric theory can be constructed in natural fashion from the datum of an equivariant sheaf over~$Y$. 

From the mathematical perspective, it is helpful to think about this procedure as providing a natural free resolution with several desirable properties. The essential object of study in a classical supersymmetric field theory is the sheaf of solutions to equations of motion modulo gauge transformations; this sheaf by definition, is equipped with an action of the supersymmetry algebra. The BRST and BV formalisms provide free resolutions of the quotient by gauge symmetries and the equations of motion, respectively; the action of supersymmetry is then lifted to a homotopy action on the space of BRST or BV fields. The pure spinor superfield construction builds a much larger free resolution of the same sheaf, which is in fact free over a supermanifold. Correspondingly, the action of the supersymmetry algebra is strict. In many examples, the resolution is also multiplicative. 

It is thus natural to expect that twisting a multiplet that is constructed using the pure spinor superfield technique will be both easier than working with the usual homotopy action on component fields, and will also relate the twisted theory to the geometry of the nilpotence variety in a neighborhood of the twisting supercharge. Both expectations turn out to be correct. As such, we see the techniques developed in this paper as interesting and worthwhile in their own right, and include numerous examples of twist computations that were known previously, to illustrate the extent to which our techniques streamline the computations and to give a new geometric perspective. We rely heavily on a version of the pure spinor superfield formalism which is tailored to our needs; a full exposition at the untwisted level will appear in the forthcoming paper~\cite{EHSW}.

In particular, it is easy to sum up all of our twisting results as follows: The component fields of a multiplet in the pure spinor formalism consist of the associated bundle to the $\Spin(d)$-representation on the Koszul homology of (the corresponding equivariant sheaf over) the nilpotence variety $Y$. To study the holomorphic twist, we consider the cohomology of the super-Poincar\'e algebra with respect to the differential $[Q,-]$. The component fields of the twisted theory consist of the associated \emph{holomorphic} bundle to the $\U(n)$-representation on the Koszul homology of the nilpotence variety~$Y_Q$ of this twisted super-Lie algebra. The BRST or BV differentials that are present are induced, at the cochain level, by a further term in the differential that involves the structure constants of the relevant Lie algebra (either twisted or untwisted).

\subsection*{Note added}
We thank R.~Eager and F.~Hahner for coordinating the submission of their related paper~\cite{EagerHahner}. In that work, the authors compute the \emph{maximal} twist of free perturbative eleven-dimensional supergravity, starting with a component-field description that draws on the work of~\cite{Ced-11d}. The form of the maximal twist was conjectured by Costello in~\cite{CostelloMtheory2}, and has been derived from the holomorphic twist presented here in unpublished work by Surya Raghavendran.

\subsection*{Acknowledgements}

Special thanks are due to R.~Eager, F.~Hahner, and J.~Walcher, as well as to S.~Raghavendran, for close collaboration on related topics, both past, present, and future.
We would also like to acknowledge K.~Costello and Si Li's groundbreaking work on string field theories associated to topological string theories~\cite{CLsugra, CLbcov1,CLtypeI}.
This work led to the conjectural descriptions of twists of ten-dimensional supergravity to which we compare in the last section of this paper. 
We further thank I.~Brunner, D.~Butson, M.~Cederwall, T.~Dimofte, C.~Elliott, O.~Gwilliam, J.~Huerta, J.~Palmkvist, N.~Paquette, P.~Safronov, P.~Yoo for conversation and inspiration of all kinds. %I.S.\ thanks the Fields Institute at the University of Toronto for hospitality, as well as the Mathematical Sciences Research Institute in Berkeley, California and the Perimeter Institute for Theoretical Physics for generous offers of hospitality that did not take place due to Covid-19. The work of I.S. is supported in part by the Deutsche Forschungsgemeinschaft, within the framework of the cluster of excellence ``STRUCTURES'' at the University of Heidelberg. 
The work of I.S.\ is supported by the Free State of Bavaria.
The work of B.W.\ is supported by the University of Edinburgh. 
%The work of B.R.W. is supported by the National Science Foundation Award DMS-1645877 and by the National Science Foundation under Grant No. 1440140, while the author was in residence at the Mathematical Sciences Research Institute in Berkeley, California, during the semester of Spring 2020.

\subsection*{Notation}
We will write $\lie{g}$ either for a super Lie algebra or a (differential) $\Z$-graded Lie algebra, according to context. For super Lie algebras, $\lie{g}_\pm$ will refer to the odd and even elements; we write $\lie{g}_{\geq 0}$ for the positively graded subalgebra of a $\Z$-graded algebra.

\section{Background}
\subsection{Twisting}\label{sec:twisting}
We briefly recall the notion of a \emph{twist} of a representation of a supersymmetry algebra, mostly to fix notation. 
For a more complete introduction to this topic, the reader is referred to the literature, we most closely follow \cite{CostelloHolomorphic,ESW}.

A field theory on affine space admits an action of the affine (or ``Poincar\'e'') group. 
A theory is \emph{supersymmetric} when it admits an action of a super Lie algebra $\lie{p} = \lie{p}_+ \oplus \Pi \lie{p}_-$ containing an affine subalgebra, which extends the standard action of that subalgebra. 
The construction of the relevant superalgebras will be reviewed below. 

For us, a \emph{supermultiplet} on affine space is a $\Z$-graded chain complex of $\Z/2\Z$-graded (super) vector bundles, equipped with a homotopy $\lie{p}$-module structure. 
We think of this as a free resolution or derived replacement of some (non-locally-free) $\Z/2\Z$-graded sheaf, on which the supersymmetry algebra acts. 
Note that we do not specify \emph{which} physical object this sheaf plays: it may consist of all field configurations; the space of gauge-equivalence classes of field configurations; or the space of solutions to the equations of motion, considered up to gauge symmetry. 
In the last case, one normally says that the supersymmetry acts ``on-shell,'' as contrasted with the ``off-shell'' action of the other cases. 

Given a supermultiplet $E$ (which is a $\ZZ \times \ZZ/2$ graded object) and an odd element $Q \in \lie{p}_-$ such that $[Q,Q]= 0$, the \emph{twist} $E_Q$ of~$E$ by~$Q$ is defined to be the family 
\[
E[u] = E \otimes_\CC \CC[u] .
\]
over the line $\Spec(\CC[u])$, equipped with the differential $(\d_E + uQ)$.
We equip the polynomial algebra $\CC[u]$ with the structure of a $\ZZ$-graded super algebra where $u$ has bidegree $(1,-)$. 
In turn, $E_Q$ inherits the structure of a $\ZZ \times \ZZ/2$ graded cochain complex. 

One can often regrade using a \emph{twisting datum} to place $u$ in degree $(0,+)$, so that it makes sense to specialize $u = 1$ and view the fiber of this family itself as a $\Z\times \Z/2\Z$-graded chain complex with differential $(d_E + Q)$. For details on this, we refer to the literature~\cite{CostelloHolomorphic,ESW}. By construction, the twisted  multiplet~$E_Q$ is equipped with a homotopy action of the ``twisted'' differential superalgebra 
\deq{
\lie{p}_Q = (\lie{p}[u],u\ad_Q).
}
\begin{rmk}
\label{rmk:loc}
The algebraic family $E_Q$ over the formal disk contains the untwisted multiplet $E$ as the fiber over zero. In the physics literature, the ``twist'' more properly refers to the $Q$-deformed object that is the generic fiber of this family. As such, in considering the twisted multiplet, we will frequently localize to the family $(u)^{-1} E_Q$ over the \emph{punctured} disk $\Spec \C[u,u^{-1}]$. While the reader should remain attentive, no confusion should arise.
\end{rmk}

\subsection{The nilpotence variety}
For any super Lie algebra $\fg = \fg_+ \oplus \Pi \fg_-$, we consider the space $Y$ of odd elements whose self-bracket is zero:
\deq{
Y = \{Q \in \lie{g}_- \; | \;  [Q,Q] = 0 \}.
}
This space is naturally an algebraic variety: for us, it will appear only through its (graded) ring $\O_Y$ of global functions, which is a quotient of the polynomial ring on $\lie{g}_-^\vee$ by an ideal generated by homogeneous quadratic equations. 
Thus $Y$ is the affine scheme $\Spec \O_Y$; we may occasionally make reference to $\Proj \O_Y$, but it should be clear from context when this is intended.

A few comments are in order. 
Firstly, $Y$ is closely related to the space of Maurer--Cartan elements of a $\ZZ$-graded Lie algebra. 
If $\lie{g}$ admits a lift of the parity to a $\Z$-grading with support in degrees 0, 1, and 2, then $Y$ is precisely the space of Maurer--Cartan elements of this lift. 
Note, though, that we do not take the quotient by gauge transformations, instead choosing to remember that $\lie{g}_0$ acts by vector fields on~$Y$ (we make this precise momentarily). 

This condition that $\fg$ lifts to a $\ZZ$ graded Lie algebra concentrated in degrees $0,1,2$ may appear arbitrary; it will apply, however, to all supertranslation algebras, and thus to every example we will study in this paper. 
In general, though, it may not be satisfied, even if $\lie{g}$ can be lifted to a $\ZZ_{\geq 0}$ graded Lie algebra. $Y$ is then not precisely the space of Maurer--Cartan elements, but it does naturally embed in $Y$ in a way which is equivariant for the $\lie{g}_0$ action.

Secondly, it should be apparent that, if $\lie{g} = \lie{p}$ is a supersymmetry algebra, the space $Y$ precisely classifies the possible twists of a multiplet with $\lie{p}$-supersymmetry, in the sense defined above.

As such, the possible twists of a $\lie{p}$-multiplet (indeed, any $\lie{p}$-module) form a natural family over~$Y$, and it is natural to ask about performing the twisting construction in families. 
It turns out that this is the essential idea behind the ``pure spinor superfield'' approach to constructing (untwisted!) $\lie{p}$-multiplets; as we show in the sequel, the structure of such multiplets, as well as the computation of their twists, are related intimately to the algebraic geometry of~$Y$.

\subsection{Pure spinor superfields}

We now give a quick overview of the pure spinor superfield construction. 
This technique has been developed for quite some time in the physics literature, notably in work of Cederwall et. al.~\cite{Cederwall,CederwallM5}, Berkovits~\cite{BerkovitsCovariant,Berkovits1,BerkovitsNM}, Movshev and Schwarz \cite{MovshevSchwarz,MovshevSchwarz2}, and references therein. 
More recently, it was reinterpreted and generalized in~\cite{NV}. 
For a complete and modern treatment in sheaf-theoretic language, we refer to the forthcoming work~\cite{EHSW}; here, we just briefly summarize those things that are needed for our purposes.

From this point forward, we restrict to the following setting: We let $T$ be a super Lie group with super Lie algebra $\lie{t}$, and $T_+$ the underlying (even) Lie group whose Lie algebra is~$\lie{t}_+$. The space $C^\infty(T)$ of smooth functions on~$T$ has two commuting actions of $\lie{t}$ by vector fields, on the left and on the right.
(The letter refers to ``(super)translations,'' although the construction is still somewhat more general at this point.) We further assume that $\lie{t}$ admits a strictly positive $\Z$-graded lift. Let $\lie{p}$ be the nonnegatively graded Lie algebra obtained by extending $\lie{t}$ by its automorphisms in degree zero; thus $\lie{p}_0 = \aut(\lie{t})$ and $\lie{p}_{>0} = \lie{t}$. ($\lie{t}$ clearly has no degree-preserving inner automorphisms.)

We will work in the context of (differential) algebras that are graded-commutative with respect to a bigrading by $\Z\times \Z/2\Z$. Such objects will be called \emph{commutative (differential) graded superalgebras}, or c(d)gsa's. The differential of a cdgsa is required to be a square-zero derivation of degree $(1,+)$. 
In the case that the algebra is given by the sections of a $\ZZ \times \Z/2 \Z$ graded vector bundle on a manifold where the differential and product are differential and bi-differential operators, respectively, we refer to this as a {\em local} cdgsa.
We begin by exhibiting a cgsa structure on the ingredients of the construction.

The $\Z/2$-graded algebra $C^\infty(T)$ of smooth functions on the super Lie group can be presented as
\deq{
C^\infty(T) \cong C^\infty(T_+) \otimes \wedge^\bu\left(\lie{t}_-^\vee\right) .
}
It acquires the structure of a cgsa in a trivial way, by viewing the generators $\lie{t}_-^*$ of the exterior algebra as placed in bidegree $(0,-)$.
Via this presentation, we see that $C^\infty(T)$ is naturally a filtered cgsa, with the decreasing filtration $F_1$ defined by 
\deq{
F_1^k C^\infty(T) = C^\infty(T_+) \otimes \wedge^{\geq k} \left(\lie{t}_-^\vee\right).
}
(Indeed, this filtration arises from the natural $\ZZ$-graded algebra structure on the exterior algebra; for our purposes, though, it will be natural to use the cgsa structure defined above.)

Since $Y$ is defined by a system of quadratic equations in the affine space $\lie{t}_-$, $\O_Y$ has the structure of a commutative, $\Z$-graded algebra. 
While this is not a graded-commutative algebra structure, $\O_Y$ can be viewed as a cgsa by declaring that the linear functions are of bidegree $(1,-)$. Of course, $\O_Y$ also carries the decreasing filtration associated to its integer grading, which we will call $F_2$.

\subsubsection{The construction}

We consider the tensor product
\deq{
A = C^\infty(T) \otimes_\C \O_Y.
}
Since it is the tensor product of cgsas, $A$ is a cgsa in a natural fashion---either over $\C$, or over the ground ring $C^\infty(T_+)$.

We recall that there are commuting right and left actions of~$\lie{t}$ on~$C^\infty(T)$, given by super Lie algebra maps
\deq{
R, L \colon \lie{t} \to \Vect\left(T\right)
}
where $\Vect(T)$ is the super Lie algebra of super vector fields.
Since $A$ is a module over itself, there is also a map
\deq{
m\colon \lie{t}_-^\vee \hookrightarrow A \rightarrow \End_{A\operatorname{-}\Mod}(A),
}
which is just multiplication by linear generators.

We observe that $A$ is equipped with the following structures:
\begin{enumerate}[label=\emph{\alph*}., labelsep = *, leftmargin = \parindent, labelindent = 0 pt]
\item There is a $\lie{p}_0$ action on~$A$, obtained as the tensor product of the $\lie{p}_0$-module structures on $C^\infty(T)$ and~$\O_Y$.
\item There is a $\lie{t}$-action on $A$, obtained via taking the tensor product of the vector fields $L$ with the identity automorphism of~$\O_Y$. Together, these two data define the structure of a $\lie{p}$-module on~$A$.
\item There is a natural differential $\cD$ of bidegree $(1,-)$ on~$A$, obtained as the image of the natural map
\deq{
1 \mapsto \lie{t}_- \otimes \lie{t}_-^\vee \xrightarrow {R \otimes m} \Der_\C(A).
}
$\cD$ is a derivation, is $\lie{p}_0$-invariant, and squares to zero. It therefore gives $A$ the structure of a cdgsa.
\item There is a natural filtration $F$ on $A$, defined by
\deq{
F^k A = \bigoplus_{i + j = k} F_1^i C^\infty(T) \otimes_\C F_2^j \O_Y.
}
The differential $\cD$ respects this filtration.
\end{enumerate}

In sum, $A$ has the structure of a $\lie{p}$-equivariant filtered cdgsa. We will refer to $A$ as the \emph{tautological (filtered) cdgsa} associated to the data $(T,\lie{t})$. 

Physically speaking, $A$ defines (a large resolution of) a particular multiplet on the affine spacetime $\A(T_+)$ with an action of the super-Poincar\'e algebra $\lie{p}$. This resolution has several crucial properties: It strictifies the homotopy action of $\lie{p}$ on the multiplet, which is almost never strict in the typical (``component-field'') resolutions used in physics. Relatedly, it is a free resolution over the \emph{affine superspace} $\A(T)$, not merely over~$\A(T_+)$. Last but not least, it is a \emph{multiplicative} free resolution, since $A$ has the structure of an algebra over~$C^\infty(T)$; as should be clear, the structure of this multiplicative free resolution is controlled just by the classical algebraic geometry of the affine scheme $\Spec(\O_Y)$. From our perspective, these are the essential structural properties of the pure spinor superfield formalism. In what follows, we will see that twists of multiplets that are constructed in this way are also governed in a pleasing fashion by the algebraic geometry of~$Y$.

\begin{rmk}
As observed in~\cite{NV} and used implicitly in previous literature, this construction can straightforwardly be generalized: 
Let $\cE$ be a graded $\O_Y$-module that is equivariant for the action of~$\lie{p}_0$. (We think of this as the space of global sections of an equivariant sheaf on the affine scheme $Y$.) The tensor product
$A_\cE = C^\infty(T) \otimes_\C \cE$  can similarly be equipped with a tautological differential, and produces a multiplet for $\lie{p}$. However, if $\cE$ does not carry a commutative algebra structure, there is no natural algebra structure on~$A_\cE$. We will not consider examples of this form here; the reader is referred to~\cite{EHSW}, or to the prior literature.
\end{rmk}

\subsubsection{The role of the filtration}
Recall that the \emph{Koszul complex} of the graded ring $\O_Y$ is defined to be the Koszul complex of its maximal ideal:
\deq{
K^\bu(\O_Y) = \bigg( \O_Y \otimes_\C \wedge^\bu(\lie{t}^\vee_-), ~ \d_K \bigg),
}
where $\d_K$ is the Koszul differential induced by the shift morphism from $\lie{t}^\vee_-$ to~$\lie{t}^\vee_-[-1] \subset \O_Y$. Note that the Koszul differential has bidegree $(1,+)$. In coordinates, if $\theta$ refers to generators of bidegree $(0,-)$, $\lambda$ to generators of bidegree $(1,-)$, and $I$ to the defining ideal of~$\O_Y$, then
\deq{
K^\bu(\O_Y) = \bigg( \left( \C[\lambda]/I \right)  \otimes_\C \C[\theta] \; , \; ~ \d_K = \lambda \pdv{}{\theta} \bigg).
}
We will often write $K^\bu(Y)$ as a shorthand for $K^\bu(\O_Y)$. 
We remark that these choices of degrees for generators endow $K^\bu(Y)$ with the structure of a cdgsa. 
Moreover, there is a compatible filtration on $K^\bu(Y)$ which assigns weight one to both $\lambda$ and $\theta$ (so that the differential $\d_K$ is weight zero).

\begin{fact}
The associated graded $\Gr A$ of the tautological filtered cdgsa~$A$ is freely generated over $\A(T_+)$ by the Koszul complex $K^\bu(Y)$: that is,
\deq{
\Gr A \cong C^\infty(T_+) \otimes_\C K^\bu(Y).
}
\end{fact}
Indeed, $A$ can be viewed as a deformation of $C^\infty(T_+) \otimes_\C K^\bu(Y)$ by a new differential that is a first-order differential operator on~$\A(T_+)$. In coordinates, this new term in the differential looks like $ \lambda \theta \cdot \partial/\partial x$, where the relevant contraction is defined by Clifford multiplication (indeed, speaking more generally, by the bracket of~$\lie{t}$). In examples, the cohomology of~$K^\bu(Y)$ acquires the interpretation of the ``component fields''\footnote{More precisely, it describes the bundle in which the fields take values.} of the multiplet; further differentials in the spectral sequence of this filtration correspond, in physical terms, to the BRST or BV differentials of the multiplet. For a more precise explanation in terms of homotopy transfer, we refer the reader to~\cite{EHSW}.

\subsubsection{Further mathematical motivation}
As observed in~\cite{NV}, $\O_Y$ is closely related to the Lie algebra cohomology of the supertranslation algebra.
This relationship suggests that the tautological filtered complex is closely related to other, perhaps more well-known constructions; we include a few brief remarks here as motivation for mathematically-minded readers.

Given a positively graded Lie algebra $\lie{t} = \lie{t}_{>0}$, we can consider the corresponding formal moduli problem, which is also called the `classifying space' $B\lie{t}$. 
This space is defined to be the affine dg scheme whose ring of functions is the Chevalley--Eilenberg cochain complex $\clie^\bu(\lie{t})$ that computes Lie algebra cohomology. 
As motivation for the construction we will use, recall that a $\lie{t}$-module $M$ corresponds naturally to a sheaf $\cM$ on $B\lie{t}$ via a generalization of the associated bundle construction. 
As a $\clie^\bu(\lie{t})$-module, this is simply the cochain complex $\clie^\bu(\lie{t} , M)$ computing the Lie algebra cohomology with coefficients. 
There is a filtration on this cochain complex whose associated graded is $\clie^\bu (\ft) \otimes M$. 
Algebraically, this corresponds to viewing $M$ as a trivial $\ft$-module. 

For supertranslation algebras, one has $H^0(\lie{t}) = \O_Y$. 
Consider the spectral sequence associated to the filtration on $\clie^\bu(\ft, M)$. 
On the $E_1$ page, the degree-zero slice (with respect to the horizontal grading) is precisely the tautological algebra $A$ as a filtered cgsa.
The differential on the $E_1$ page is the tautological differential $\cD$ defined above. 

This immediately suggests that $\clie^\bu(\lie{t},M)$ could be interpreted as a further derived replacement for~$A$, and furthermore that techniques of Koszul duality could be applied to better understand the scope of the pure spinor superfield construction, using the equivalence between the module categories of~$U(\lie{t})$ and~$C^\bu(\lie{t})$. We reserve further study of such considerations for future work, including~\cite{EHSW}.

\subsection{The general structure of the computation}

Fix an odd element $Q \in \fp$ for which $[Q,Q]=0$. We can then define the twisted tautological cdgsa $A_Q$. 
As a cdgsa this is
\[
A_Q = \bigg(A [u] \; , \; \cD + u \cQ \bigg) 
\]
As discussed in~\S\ref{sec:twisting}, we can view $\CC[u]$ as a cdgsa (with zero differential) where $u$ has bidegree $(1, -)$. 
Further, the integer grading determines the structure of a filtered cdgsa on $\CC[u]$; in turn this determines a filtration on the tensor product $A[u] = A \otimes_\CC \CC[u]$. 
With this filtration, the differential $\cD + u \cQ$ endows $A_Q$ with the structure of a filtered cdgsa.

\subsubsection{Preliminaries on complex geometry}
\label{sec:dol}
We consider the commutative differential graded algebra of Dolbeault forms on a complex manifold $X$. 
This can be extended to a family over the affine line by considering the algebra 
\deq{
\Dol(X) = \bigg(  \Omega^{0,\bu}(X)[u]\; ,\; ~ u \dbar \bigg).
}
(Note that the fiber over the origin is the space of Dolbeault forms equipped with \emph{zero} differential; as in Remark~\ref{rmk:loc}, we will usually localize to the family $(u)^{-1} \Dol(X)$ over the punctured affine line.)

By a slight abuse of notation, we will also use this notation for odd-dimensional manifolds equipped with transverse holomorphic foliations (THF structures). The only relevant example in this paper will just be $X \times \R$, with $X$ a complex manifold. By definition, in this case,
\deq{
\Dol(X \times \R) =  \bigg( \Omega^{0,\bu}(X) \, \Hat{\otimes}_\pi \, \Omega^\bu(\RR) [u] \; , \; ~ u \left( 1 \otimes \d +  \dbar \otimes 1 \right) \bigg).
}
The notation is motivated by the fact that this sheaf plays the same role in the minimal twist of odd-dimensional multiplets that the sheaf of holomorphic functions (or rather its Dolbeault resolution) does in the minimal twist of even-dimensional multiplets.

We give $\Dol(X)$ the structure of a cdgsa by assigning bidegree $(1,-)$ to the even variable $u$, and bidegree $(0,-)$ to $\d\zbar$. The differential $u\dbar$ then has bidegree $(1,+)$ as required. 
We can further equip $\Dol(X)$ with a filtration $F$ by setting 
\deq{
F^k \Dol(X) = \bigoplus_{i + j = k} \Omega^{0,\geq i}(X) \otimes u^j \C[u].
}
This gives $\Dol(X)$ the structure of a filtered (local) cdgsa on~$X$. If $X = \C^n$ is an affine space, $\Dol(X)$ is equivariant in an obvious fashion for the $\lie{gl}(n)$ action on~$X$. The generalization to the case $X \times \R$ is obvious.

In this paper, $X$ will always be flat. We will often leave the complex structure implicit, and write (for example) $\Dol(\R^d)$, where $d$ may be even or odd, and the presence of a THF structure identifying $\R^d$ with $\C^n$ or $\C^n \times \R$ is understood. No confusion should arise.

\begin{rmk}\label{rmk:regrade}
There is a way to regrade 
$\Dol(X)$
to put it in its usual $\ZZ$-graded form. 
First, notice that there is a $\CC^\times$ action on this complex given by declaring that $|\d \zbar| = +1$ and $|u| = -1$. 
The differential is clearly invariant under this action.
Let ${\rm wt}(\alpha)$ be the resulting $\CC^\times$-weight of a general element $\alpha$ and define a new $\ZZ \times \ZZ/2$-grading $|a|'$ by the rule
\[
|a|' = |a| + ({\rm wt}(a), {\rm wt}(a) \; {\rm mod\; 2}) .
\]
Then, it is clear that $|u| = (0,+)$, $\Omega^{0,\bu}(X)$ carries its ordinary $\ZZ$-grading, and the $\ZZ/2$-grading is trivial. 
If we specialize $u =1$ we obtain the ordinary $\ZZ$-graded Dolbeault complex $(\Omega^{0,\bu}(X), \dbar)$. 
\end{rmk}

\subsubsection{Holomorphic decomposition} 
The connection between supersymmetry and complex geometry is provided by twisting. Every supersymmetry algebra with at least four supercharges admits a minimal or holomorphic twist; we will abuse terminology slightly by using these terms interchangeably, thus also referring to minimal twists in odd dimensions as ``holomorphic.''

The existence of the holomorphic twist is ensured by general properties of Clifford modules. Recall that the spin representation of $\lie{spin}(V)$ can be constructed by choosing a maximal isotropic subspace $L \subset V_\C$. We can then fix a decomposition of~$V_\C$ as a direct sum,
\deq{
  V_\C = L \oplus L^\vee \oplus (L^\perp/L).
}
Note that $(L^\perp/L)$ is either trivial or one-dimensional, depending on whether $d = \dim(V)$ is even- or odd-dimensional. $L$ itself is $n = \lfloor d/2 \rfloor$-dimensional, so that $d = 2n$ or $2n+1$. The stabilizer of this decomposition is $\lie{gl}(n) \subset \lie{so}(d)$ (at least if $d$ is even).

The (``Dirac'') spin representation $S$ is then constructed by taking the exterior algebra on the linear dual~$L^\vee$ (twisted by a square root of the determinant line of~$L$, which we will often tacitly ignore, see the remark below). Clifford multiplication is determined by taking $L^\vee$ to act by wedging and $L$ by contraction; if present, $(L^\perp/L)$ acts by a sign $\pm(-1)^k$ on $k$-forms, preserving form degree. The action of the complex spin algebra is determined by the isomorphism $\lie{spin}(V) \cong \wedge^2 V_\C$. When $d$ is odd (Dynkin type $B_n$), the spin representation is irreducible; when $d$ is even (Dynkin type $D_n$), it is the direct sum of two irreducible ``chiral'' spin representations,
\deq{
  S_+ = \wedge^\text{even} L^\vee, \quad S_- = \wedge^\text{odd} L^\vee.
}
All of these constructions are $\lie{gl}(L)$-equivariant; we note that the choice of $L\subset V_\C$ determines an identification of $V_\R$ with either $\C^n$ or $\C^n \times \R$, as the case may be (up to a choice of complex basis). If $d$ is even, this is just a choice of complex structure on~$\R^d$. The real form of the stabilizer of~$L$ is~$U(n)$; the moduli space of such choices is therefore the orthogonal Grassmannian $SO(d)/U(n)$, which is a symmetric space for $SO(d)$ of complex dimension $n(n-1)/2$ or $n(n+1)/2$, for even or odd $d$ respectively.

In what follows, we will use this decomposition of the spin representation repeatedly, and will often refer to the $i$-form component of a spinor variable just with a subscript $i$. Thus, if $\lambda$ is a coordinate in a spin representation, $\lambda_i$ denotes a coordinate on the summand $\wedge^i L^\vee$.

\begin{rmk}
  These considerations are closely connected to the geometry of spin bundles on Kähler manifolds, where they are encapsulated by the bundle isomorphism 
  \deq{
    S \cong K_X^{1/2} \otimes \wedge^\bu (T^{\vee}_X)^{0,1}.
  }
Locally, the (square-root of the) canonical bundle is trivialized, and we will assume a trivialization in the remainder of the paper.
The global analogue of~$L$ is the antiholomorphic tangent bundle $T^{0,1}_X$, which acts trivially by contraction on the Dolbeault complex. The Kähler metric provides an identification of~$(T^\vee_X)^{0,1}$ with~$T^{1,0}_X$, so that one sees that the holomorphic translations do not appear in the image of~$\ad_Q$ and thus survive the twist.
\end{rmk}

We can now discuss the construction of the holomorphic twist. The form of the pairing in the supertranslation algebra (which is always constructed using Clifford multiplication) ensures that a supercharge in $\wedge^i L^\vee$ can pair nontrivially only with $\wedge^j L^\vee$ when $i+j = n-1$, $n$, or $n+1$. As such, if $n>2$, the supercharge in~$\wedge^0 L^\vee$ is always of square zero, and is determined (up to $R$-symmetry) by a choice of~$L$; conversely, a supercharge $Q$ of this sort determines $L$ by considering its annihilator under Clifford multiplication.

There is thus always a stratum in the nilpotence variety isomorphic to the orthogonal Grassmannian of maximal isotropic subspaces in~$V_\C$, or to a product of this space with another factor if $R$-symmetry is present. Supercharges in this stratum define holomorphic twists; conversely, we can construct the holomorphic twist by placing the theory on a Kähler manifold and considering its deformation with respect to a supercharge of zero-form type. 

Since $\lie{p}$ admits an integral lift of the $\Z/2\Z$ grading with support in degrees zero, one, and two, a square-zero supercharge is exactly the same thing as a Maurer--Cartan element, and thus determines a deformation of~$\lie{p}$ to a dg Lie algebra $\lie{p}_Q = (\lie{p}, \ad_Q)$. We will have cause to consider this dg Lie algebra, and in particular its minimal model $H^\bu(\lie{p}_Q)$, frequently in what follows. 

In degree zero, the cohomology is determined by a kernel condition, and consists of the stabilizer of~$Q$ in~$\fp_0$; this is the direct sum of a parabolic subgroup related to the compact subgroup $SU(n)$ with the stabilizer of~$Q$ in the $R$-symmetry algebra. Restricting to the even-dimensional case for a moment and using the decomposition
\deq{
  \lie{spin}(V) = \wedge^2V = \wedge^2 L \oplus \wedge^2 L^\vee \oplus \lie{gl}(L),
}
we see that the kernel of~$\ad_Q$ is $\lie{sl}(L) \oplus \wedge^2 L$. The computation in odd dimensions is similar.
  Similarly, in degree two, the cohomology is the cokernel of $\ad_Q$, which can be identified with the space $L$ of surviving (holomorphic) translations. 

  In degree one, the computation is slightly more involved. Forms of degree $(n-1)$ and $n$ generally fail to commute with $Q$ (unless they do for reasons of $R$-symmetry), and so are eliminated by the kernel condition. In fact, the subspace of supercharges failing to commute with $Q$ is always isomorphic via $\ad_Q$ to the space of $Q$-exact translations, which it nullhomotopes in the twist. Geometrically, such supercharges correspond to the algebraic normal bundle of~$Y$ at~$Q$.
  
  The image condition nullhomotopes two-form supercharges, using the identity morphism on the $R$-symmetry space. 
  Taking $d$ to be even and ignoring $R$-symmetry for simplicity, it is easy to see that the dimension of the image is $n(n-1)/2$, which agrees with the dimension of the orthogonal Grassmannian as computed above. Similar computations, which we will show in detail in examples below, show that the image of $\ad_Q$ in degree one is the tangent space at~$Q$ to the lowest stratum of holomorphic supercharges in~$Y$. 

  The odd elements of $H^\bu(\lie{p}_Q)$ thus correspond to deformations of the holomorphic supercharge $Q$, modulo deformations to other holomorphic supercharges, which obey the square-zero condition at linear order. Furthermore, some brackets between odd elements can, and do, survive. There is thus a nilpotence variety $Y_Q$ associated to~$H^\bu(\lie{p}_Q)_{>0}$, which reflects the fact that the holomorphic supercharges lie on a singular locus of the nilpotence variety $Y$, and whose defining equations represent the obstructions to extending first-order deformations of $Q$ off of the holomorphic stratum to genuine deformations. (Of course $Y_Q$ is also the naive Maurer--Cartan space of~$\lie{p}_Q$.) $Y_Q$ will play an important role in what follows.

\subsubsection{The structure of the proofs}
We now sketch the form of the results and outline the computational techniques that are used to obtain them. These are uniform in each case. As such, we will see the pattern sketched here played out in each example in the remainder of the paper.

We always begin with the tautological filtered cdgsa $A$, representing a specific multiplet for a chosen supersymmetry algebra in $d$ real dimensions.\footnote{As emphasized above, other equivariant sheaves on $\O_Y$ would give rise to other multiplets, though we do not consider this here. The tautological filtered cdgsa will produce either a vector, tensor, or gravity multiplet, depending on dimension.} In~\S\ref{sec:smooth}, the relevant algebras are for minimal supersymmetry in dimensions four, six, and ten; we treat these uniformly. The key unifying feature is that the nilpotence variety $Y$ is smooth in each case; as we will see below, this implies that the twist of the multiplet is just given by $\Dol(\RR^d)$. Further sections treat the abelian $\N=(2,0)$ multiplet in six dimensions, eleven-dimensional supergravity, and finally the type IIB supergravity multiplet in ten dimensions.

\begin{enumerate}[label=\emph{\alph*}., labelsep = *, leftmargin = \parindent, labelindent = 0 pt]
  \item The first step is to consider the twist of the tautological filtered cdgsa; as recalled above, this is given by
    \deq{
      A^\bu_Q = \left( A^\bu[u], \cD + u \cQ \right),
    }
    which acquires the structure of a filtered cdgsa by placing $u$ in bidegree $(1,-)$. Note, however, that $A^\bu_Q$ is no longer $\lie{p}$-equivariant; rather, it is equivariant for the differential graded super Lie algebra $\lie{p}_Q$. 
    \item We will use the spectral sequence associated to the filtration $F^\bu A_Q$ to study the cohomology of the twisted cdgsa $A_Q$. 
Recall from above that the associated graded of the (untwisted) tautological cdgsa $A$ is $\Gr A \cong \rC^\infty(\RR^d) \otimes K^\bu (Y)$
where $K^\bu(Y)$ is the Koszul complex associated to the nilpotence variety. 
For $Q$ a holomorphic supercharge it's immediate to see the associated graded of the twisted cdgsa $A_Q$ is
\[
\Gr A_Q \cong \bigg(\rC^\infty(\RR^d) \otimes K^\bu (Y) [u] \; , \; \d_K + u \frac{\partial}{\partial \theta_n} \bigg) 
\]
where $\d_K$ is the (untwisted) Koszul differential acting on $K^\bu(Y)$. 
\item 
After inverting $u$ we will identify the cohomology of this associated graded
\[
H^\bu\big( (u)^{-1} \Gr A_Q \big) \cong (u)^{-1} \Dol(\RR^d) \otimes H^\bu \big(K^\bu(Y_Q)\big) 
\]
as cgsa's (trivial differentials). 
\item This completely describes the $E_2$ page of the spectral sequence computing the cohomology of $(u)^{-1} A_Q$.
In all of our examples the sequence collapses at $E_3$, so it suffices to characterize the differential on this page which consists of the differential acting on $\Dol(\RR^d)$ together with a differential induced by the Lie bracket in the algebra $\fp_Q$.

\end{enumerate}

As part (d) indicates the structure of the graded Lie algebra $H^\bu(\lie{p}_Q)$ and the geometry of the nilpotence variety $Y$ in the neighborhood of the point corresponding to the holomorphic supercharge $Q$ are closely related. 
Indeed, after restricting to the complement of a hyperplane section determined by the choice of holomorphic supercharge $Q$, $Y$ is the total space of a fibration over $Y_Q$, whose fibers are affine and correspond to the stratum of holomorphic supercharges. 
In the case where the nilpotence variety is smooth, the algebra $\fp_Q$ is abelian, and one is simply left with the differential on the Dolbeault complex.

\section{Smooth varieties and Yang--Mills multiplets}
\label{sec:smooth}

\subsection{Details on minimal supersymmetry in dimensions four, six, and ten}
\label{ssec:mindetails}

In this section, we quickly recall a few standard facts about the relevant supersymmetry algebras, their nilpotence varieties $Y$, their twists, and their associated tautological filtered algebras. 
For more details, the reader is referred to~\cite{BerkovitsCovariant, Cederwall, MovshevSchwarz,NV}.

\subsubsection{$\N=1$ supersymmetry in four dimensions}

The $\ZZ$-graded lift of the (complexified) four-dimensional supertranslation algebra takes the form
\deq{
  \lie{t} =  (S_+ \oplus S_-)[-1] \oplus V[-2].
}
The automorphisms are $\lie{p}_0 = \so(4) \oplus \lie{u}(1)$; $V \cong \C^4$ denotes the fundamental representation of $\so(4,\C)$ and $S_\pm$ the chiral spin representations. The $\lie{u}(1)$ weights are determined by the chirality.

There is a $\lie{p}_0$-equivariant isomorphism
\deq[eq:4dbracket]{
  S_+ \otimes S_- \cong V.
}
The bracket is just defined by extending this map by zero.  Since~\eqref{eq:4dbracket} is an isomorphism, it is trivial to see that the variety $Y$ consists of the two hyperplanes $S_+$ and $S_-$ inside of~$S_+ \oplus S_-$, intersecting transversely at the origin.

Fix a point $Q \in S_+ \subset \NV{4}$; such a choice determines a complex structure on $V_\R$, and in turn a maximal isotropic subspace $L \subset V$. As reviewed above, there is then a decomposition
\deq{
  S_+ \oplus S_- = \wedge^\bu(L^\vee),
}
where the chirality corresponds to the parity of the form degree.
We will write $\lambda_i$ for coordinates in a corresponding decomposition, where $i \in \{0,1,2\}$. Using these coordinates, the defining equations of~\NV{4} become 
\begin{equation}\label{eqn:4d1ps}
  \lambda_0 \lambda_1 = 0, \qquad  \lambda_2 \lambda_1  =  0.
\end{equation}
It is then easy to see the structure of the dg Lie algebra $\lie{t}_Q = (\ft, [Q,-])$.

\begin{prop}
As a cochain complex, the dg Lie algebra $\fp_Q$ is
\[
  \begin{tikzcd}[row sep = 1 ex]
    0 & 1 & 2 \\ \hline \\
    \wedge^2 L^\vee \ar[r] & \wedge^2 L^\vee \\
    \lie{sl}(L)  & \wedge^1 L^\vee \ar[r] & L \\
    \C \ar[r] & \wedge^0 L^\vee & L^\vee \\
    \wedge^2 L
  \end{tikzcd}
\]
where all arrows are identity morphisms.
Thus $H^\bu(\lie{p}_Q) = \lie{sl}(L) \oplus \wedge^2 L \oplus L^\vee[-2]$, and $Y_Q = \Spec \C$ is trivial.
\end{prop}

Let us spell out the tautological filtered cdgsa using these coordinates. 
As a commutative graded super algebra
\deq{
  A = \rC^\infty(\C^2) [\theta_0,\theta_1,\theta_2] \otimes_\C \left( \frac{\C[\lambda_0,\lambda_1,\lambda_2]}{I} \right).
}
The tautological differential takes the form
\deq{
  \cD = \lambda_2 \left( \pdv{ }{\theta_2} - \theta_1 \pdv{ }{\zbar} \right) + \lambda_1 \left( \pdv{ }{\theta_1} - \theta_2 \pdv{  }{\zbar} -\theta_0 \pdv{ }{z} \right) + \lambda_0 \left( \pdv{ }{\theta_0} - \theta_1 \pdv{ }{z} \right) . 
}
The element $Q$ acts on this algebra via the operator
\deq{
  \cQ = \pdv{ }{\theta_2} + \theta_1 \pdv{ }{\zbar}.
}
This defines the twisted tautological cdgsa $A_Q = (A[u], \cD + u \cQ)$. We further recall that $A_Q$ has a filtration defined by
\[
F^\ell A_Q = \bigoplus_{i + j + k \geq \ell} \rC^\infty (\CC^2)[\theta_0,\theta_1,\theta_2]^i \otimes \left(\frac{\CC[\lambda_0,\lambda_1,\lambda_2]^j}{I}\right) \otimes \CC[u]^k  .
\]

\subsubsection{$\N=(1,0)$ supersymmetry in six dimensions}

The six-dimensional $\N=(1,0)$ supertranslation algebra is
\[
  \ft = (S_{+} \otimes R)[-1] \oplus V[-2].
\]
Its automorphisms are $\lie{p}_0 = \so(6) \oplus \sp(1)$; $V \cong \CC^6$ denotes the fundamental representation of $\so(6, \CC)$, $S_{+} \cong \CC^4$ is an irreducible semi-spin representation, and $R$ is a two-dimensional complex vector space equipped with a linear symplectic form $\omega_{R}$, forming the defining representation of~$\lie{sp}(1)$.

There is an $\so(6, \CC)$ equivariant isomorphism
\deq{
  \Gamma_{\Omega^1} \colon \wedge^2(S_{+}) \xto{\cong} V_6 .
}
The bracket on~$\ft_{6}$ is defined by the tensor product of this map with the symplectic form $\omega_{R}$, which defines an $\so(6) \oplus \sp(1)$-equivariant map from the symmetric square of $S_{+} \otimes R$ to $V$:
\deq{
[s_+ \otimes r, s_+' \otimes r'] = \Gamma_{\Omega^1} (s_+ , s_+') \otimes \omega_R (r,r') .
}

Viewing $S_+\otimes R$ as a space of $2\times 4$ complex matrices, it is easy to see that the bracket map defines the six $2\times 2$ minors, so that the nilpotence variety $Y$ is the space of decomposable (rank-one) elements of $S_+ \otimes R$.

Now fix a point $Q \in Y_{6}$; the choice of such a point determines both a maximal isotropic subspace $L\subset V = \C^6$, corresponding to a choice of complex structure on~$\RR^6$, and a maximal isotropic subspace (in other words, a line) $\rho \subset R$. 

Using this data, we can define the following coordinates in a neighborhood of~$Q$.
As an $\lie{sl}(3)$ module, the semi-spin representation decomposes as 
\deq{
  S_{+} = \wedge^\text{even}(L^\vee) = \CC \oplus \wedge^2 L^\vee . 
}
The coordinates $\lambda$ of~$Y_6$ transform in the contragredient representation, $S_{-} \otimes R$, where $S_{-}$ is the other irreducible semi-spin representation.
Let $(\lambda_1^+, \lambda_3^+ ; \lambda_1^- , \lambda_3^-)$ be $\fgl(L) \times \fgl(\rho)$-equivariant coordinates.
The defining equations of~$Y$ then become 
\begin{equation} \label{eqn:6d10ps}
  \lambda^+_1 \lambda^-_3 - \lambda^-_1 \lambda^+_3  = 0 , \qquad
  \lambda^+_1 \wedge \lambda^-_1  = 0 .
\end{equation}
It is easy to understand the dg Lie algebra $\lie{p}_Q$ from here.

\begin{prop}
As a cochain complex, the dg Lie algebra $\fp_Q$ is
\begin{equation}
  \begin{tikzcd}[row sep = 1 ex]
    0 & 1 & 2 \\ \hline \\
    \wedge^2 L^\vee \ar[r] & (\wedge^2 L^\vee)_+ \\
    \lie{sl}(L)  & (\wedge^2 L^\vee)_- \ar[r] & L\\
    \C \ar[r] & (\wedge^0 L^\vee)_+ & L^\vee \\
    \lie{gl}(\rho) \ar[ru] \\
  (\rho^\vee)^{\otimes2} \ar[r] & (\wedge^0 L^\vee)_- \\
    \wedge^2 L \oplus \rho^{\otimes 2}.
  \end{tikzcd}
\end{equation}
Thus $H^\bu(\lie{p}_Q) = \lie{sl}(L) \oplus \wedge^2 L \oplus \rho^{\otimes 2} \oplus \lie{u}(1) \oplus L^\vee[-2]$. 
The action of $\lie{u}(1)$ is determined by the kernel condition; it is trivial on $L^\vee$, and acts with charge $-1$ on $\wedge^2 L$ and charge $+3$ on~$\rho^{\otimes 2}$. 
There are no odd elements, so that $Y_Q = \Spec \C$ is again trivial. 
\end{prop}

Let us spell out the tautological filtered cdgsa using these coordinates. 
As a commutative graded super algebra 
\deq{
  A = \rC^\infty (\CC^3)[\theta^+_1, \theta^+_3, \theta^-_1,\theta^-_3] \otimes \left(\frac{\CC[\lambda^+_1, \lambda^+_3 , \lambda^-_1,\lambda^-_3]}{I}\right) .
}
The tautological differential is
\begin{multline}
  \cD  = - \lambda^+_3 \left(\frac{\partial}{\partial \theta^+_3} - \theta^-_1 \frac{\partial}{\partial \zbar} \right) - \lambda^+_1 \left(\frac{\partial}{\partial \theta^+_1} - \theta^-_3 \frac{\partial}{\partial \zbar} - \theta^-_1 \frac{\partial}{\partial z} \right) 
   + \lambda^-_3 \left(\frac{\partial}{\partial \theta^-_3} - \theta^+_1 \frac{\partial}{\partial \zbar} \right) \\
   + \lambda^-_1 \left(\frac{\partial}{\partial \theta^-_1} - \theta^+_3 \frac{\partial}{\partial \zbar} - \theta^+_1 \frac{\partial}{\partial z} \right) .
 \end{multline}
The action of the element $Q$ 
is through the operator
\deq{
    \cQ = \frac{\partial}{\partial \theta^-_3} + \theta^+_1 \frac{\partial}{\partial \zbar} .
}
\subsubsection{$\cN=(1,0)$ supersymmetry in ten dimensions}

The super Lie algebra of the 10d $\N=(1,0)$ supertranslation group %$T_{10}$ 
is of the form
\deq{
  \lie{t} = S_+[-1] \oplus V[-2],
}
where $S_+$ denotes the chiral spin representation and $V\cong \C^{10}$ the (complex) vector representation of~$\so(10,\C)$. 
The Lie bracket is defined by the unique $\so(10,\C)$-equivariant projection map from the symmetric square of $S_+$ to~$V$. (The other irreducible component of $\Sym^2(S_+)$ is the ``pure spinor'' piece with Dynkin label $[00002]$.)

%We let $Y_{10}$ denote the corresponding nilpotence variety.
%\brian{we don't use projective above, should we change this?}As a projective variety, it can be identified with the space of maximal isotropic subspaces $L \subset V$, or equivalently with the space of complex structures on~$V_\R$, which is the symmetric space $\OG(5,10) = \SO(10)/U(5)$. The space of square-zero supercharges is the affine cone over this space. 

As is by now familiar, fixing a point $Q \in Y$ uniquely specifies a maximal isotropic subspace $L \subset V$. %which can be identified with the space of those translations that are nullhomotopic with respect to~$Q$. 
The spinor representation then decomposes as a $\lie{sl}(5)$ representation as 
\deq{
  S_+ = \wedge^\text{even} L^\vee = \CC \oplus \wedge^2 L^\vee \oplus \wedge^4 L^\vee %\otimes \det(L)^{1/2},
}
which can be identified with the space of constant evenly-graded holomorphic de Rham forms on~$\CC^{5}$ (twisted by a square root of the canonical bundle). %twisted by a square root of the canonical bundle.\footnote{After applying the twisting homomorphism appropriate to the holomorphic twist, we will identify the space of supercharges in the theory (which is isomorphic to $S_+$) just with even Dolbeault forms.} 
%\ingmar{This is an $\lie{sl}(5;\C)$ statement}

We again denote supercharges by $Q_i$, where $i$ refers to their form type. %The even part of the supertranslation algebra is $V_\C = L \oplus L^\vee$, and the corresponding elements will be denoted $P$ and $\bar{P}$. 
The brackets in the algebra are determined by Clifford multiplication, and take the form
\begin{equation}
  [Q_0, Q_4] = \bar{P}, \quad [Q_2, Q_2] = \bar{P}, \quad [Q_2, Q_4] = P.
  \label{eq:10dYMbracket}
\end{equation}
From these equations, it is clear that $Q_2$ defines the tangent space to $Y_{10}$ at~$Q_0$, and that $Q_4$ can be thought of as the fiber of the normal bundle, which is responsible for witnessing the nullhomotopy of $\bar{P}$ with respect to $Q_0$. Because $[Q_2,Q_2] = 0$ in $Q_0$-cohomology, all first-order tangent deformations are unobstructed and the variety is smooth at $Q_0$.

Let $\lambda \in S_- \cong \wedge^\text{odd} L^\vee$ be a coordinate on~$S_+$; the scalar pairing between these representations is defined by evaluation on a chosen Calabi--Yau form in $\wedge^5 L^\vee$. The bracket~\eqref{eq:10dYMbracket} then gives rise to the ten defining equations for $Y$
\begin{equation}\label{eqn:10d10ps}
  \lambda_1 \lambda_5 + \frac{1}{2} \lambda_3 \lambda_3  =  0 , \qquad
  \lambda_3 \lambda_1  = 0 .
\end{equation}

It is straightforward to understand the dg Lie algebra $\lie{p}_Q$. 

\begin{prop}
As a cochain complex, the dg Lie algebra $\ft_Q$ is 
\begin{equation}
  \begin{tikzcd}[row sep = 1 ex]
    0 & 1 & 2 \\ \hline \\ 
    \wedge^2 L^\vee \ar[r] & \wedge^2 L^\vee \\
    \lie{sl}(L)  & \wedge^4 L^\vee \ar[r] & L\\
    \C \ar[r] & \wedge^0 L^\vee & L^\vee \\
    \wedge^2 L.
  \end{tikzcd}
\end{equation}
Thus $H^\bu(\lie{p}_Q)$ consists of the parabolic Lie algebra $\lie{sl}(L) \oplus \wedge^2 L$ in degree zero, extended by the module $L^\vee[-2]$. There are no odd elements, so that $Y_Q = \Spec \C$ is once again trivial. 
\end{prop}

As in the previous examples, the property that $Y_Q$ is trivial is a straightforward consequence of the smoothness of~$Y$ at~$Q$.

Let us spell out the tautological filtered cdgsa using these coordinates. 
As a commutative graded super algebra
\deq{
  A = \rC^\infty (\CC^5)[\theta_1,\theta_3,\theta_5] \otimes \left(\frac{\CC[\lambda_1,\lambda_3,\lambda_5]}{I}\right) .
}
The explicit coordinate form of the differential is
\deq{
  \cD = \lambda_5 \left(\pdv{ }{\theta_5} - \theta_1 \pdv{ }{\zbar}\right) + \lambda_3 \left(\pdv{ }{\theta_3} - \theta_3 \pdv{ }{\zbar} - \theta_1 \pdv{ }{z}\right) + \lambda_1 \left(\pdv{ }{\theta_1} - \theta_5 \pdv{ }{\zbar} - \theta_3 \pdv{ }{z}\right) .
}

The action of the element $Q$ on $A$ is through the operator
\deq{
  \cQ = \pdv{ }{\theta_5} + \theta_1 \pdv{ }{\zbar}.
}

\subsection{Characterizing the twisted tautological complexes} \label{sec:smoothpure}

In this section we relate the twisted tautological complex $A_Q$ of minimal supersymmetry in dimensions $4,6,10$ to complex geometry. 
In \S\ref{sec:dol} we have refined the Dolbeault complex of differential forms on a complex manifold $X$ to a family over the affine line $\Dol(X)$. 
In the theorem below we localize this family over the punctured affine line.  

\begin{athm}
\label{thm:asmooth}
Suppose $d = 2n = 4$, $6$, or~$10$.
There is a $\GL(n ; \CC)$-equivariant quasi-isomorphism of filtered local cdgsa's on $\CC^{n}$
\[
    \Phi \colon (u)^{-1} \Dol(\CC^n) \xto{\simeq} (u)^{-1} A_Q .
\]
\end{athm}

We proceed to define the map $\Phi$ for each of the three cases. 
Recall that like the ordinary complex of Dolbeault forms, $\Dol(X)$ is finite rank and freely generated as a graded module over smooth functions $\rC^\infty(X)$.
We refer back to the previous subsections for notations that appear below.

\begin{itemize}
\item For $2n=4$, $\Phi$ is defined by
\[
  \begin{array}{cccl}
          & f(z, \zbar) & \mapsto & f(z - \theta_0 \theta_1 - 2 u^{-1} \lambda_0 \theta_1 \theta_2 - 2 u^{-1} \lambda_2 \theta_0 \theta_1 \; , \; \zbar + \theta_1 \theta_2) \\
          & \d \zbar & \mapsto & 2 u^{-1} \lambda_2 \theta_1 .
            \end{array}
\]
\item For $2n=6$, $\Phi$ is defined by
\[
  \begin{array}{cccl}
          & f(z, \zbar) & \mapsto & f(z + \theta^-_1 \theta^+_1 + 2 u^{-1} \theta^- (\lambda_1^- \theta_3^+ - \lambda_1^+ \theta_3^-) \; , \; \zbar + \theta^+_1 \theta^-_3 + \theta^-_1 \theta^+_3) \\
          & \d \zbar & \mapsto & 2 u^{-1} (\lambda_3^+ \theta_1^- - \lambda_3^- \theta_1^+) 
  \end{array}
\]
\item For $2n = 10$, $\Phi$ is defined by
\[
  \begin{array}{cccl}
          & f(z, \zbar) & \mapsto & f(z - \theta_1 \theta_3 - 2 u^{-1} \lambda_1 \theta_3 \theta_5 + u^{-1} \lambda_3 \theta_3^2 \; , \; \zbar + \theta_1 \theta_5) \\
          & \d \zbar & \mapsto & - 2  u^{-1} \lambda_5 \theta_1 - u^{-1} \lambda_3 \theta_3 
  \end{array}
\]
\end{itemize}

By construction, the map $\Phi$ is a map of $\ZZ \times \ZZ/2$-graded algebras. 
The following lemma is a straightforward calculation. 

\begin{alem}
In each case, $\Phi$ is a map of filtered cdgsa's.
\end{alem}
\begin{proof}[Proof for $d=4$] 
First, we record the following identities:
\begin{itemize}
\item $\cD(z) = - \lambda_1 \theta_0 - \lambda_0 \theta_1$.
\item $\cD(\theta_0 \theta_1) = \lambda_0 \theta_1 - \lambda_1\theta_0$. 
\item $\cD(\lambda_0 \theta_1 \theta_2) = \lambda_0 \lambda_1 \theta_2 - \lambda_0 \lambda_2 \theta_1$. 
\item $\cQ (\lambda_0 \theta_1 \theta_2) = -\lambda_0 \theta_1$.
\item $\cD(\lambda_2 \theta_0\theta_1) = \lambda_2 \lambda_0  \theta_1 - \lambda_2 \lambda_1 \theta_0 .$
\end{itemize}

We proceed to show $(\cD + u \cQ) \Phi(z) = 0$. 
It is clear that there are no terms of order $u$ on the left hand side of this equation. 
Furthermore, by the first, second, and fourth identities the term proportional to $u^0$ vanishes. 
Finally, using the third and fifth identities together with the pure spinor constraint \eqref{eqn:4d1ps} we see that the term proportional to $u^{-1}$ also vanishes.

Next, we consider the anti-holomorphic coordinates.
We record the following identities:
\begin{itemize}
\item $\cD(\zbar) = - \lambda_2 \theta_1 - \lambda_1 \theta_0$. 
\item $\cD (\theta_1 \theta_2) = \lambda_1 \theta_2 - \lambda_2 \theta_1$. 
\end{itemize}
Thus $(\cD + u \cQ)\Phi(\zbar) = - 2 \lambda_2 \theta_1$, which is, by definition the image of $\Phi (u \d \zbar)$ as desired. 

Finally, we observe that $\Phi$ preserves the filtrations, namely $\Phi(F^k) \subset F^k$ for each $k \geq 0$. 
Indeed, the linear generators in nonzero filtration degree are $\d \zbar \in F^1 \Dol(\CC^2)$ and $u \in F^1 \Dol(\CC^2)$.
Since $\Phi (u) = u$ and $\Phi (\d \zbar) \in F^2 A_Q \subset F^1 A_Q$, we are done.

\end{proof}

\begin{proof}[Proof for $d=6$]

First, we record the following identities:
\begin{itemize}
\item $\cD (z) = \lambda_1^+ \theta_1^- - \lambda_1^- \theta_1^+$. 
\item $\cD (\theta_1^- \theta_1^+) = \lambda_1^- \theta_1^+ + \lambda_1^+ \theta_1^-$. 
\item $\cD(\lambda_1^+ \theta_1^- \theta_3^-) = \lambda_1^+ \lambda_1^- \theta_3^- - \lambda_1^+ \lambda_3^- \theta_1^- $.
\item $\cD(\lambda_1^- \theta_1^- \theta_3^+) = \lambda_1^- \lambda_1^- \theta_3^+ + \lambda_1^- \lambda_3^+ \theta_1^-$. 
\item $\cQ (\lambda_1^+ \theta_1^- \theta_3^-) = - \lambda_1^+ \theta_1^-$.
\end{itemize}

Using these identities we show that the total differential acting on 
\[
\Phi (z) = z + \theta^-_1 \theta^+_1 + 2 u^{-1} \theta^- (\lambda_1^- \theta_3^+ - \lambda_1^+ \theta_3^-)
\]
is zero.
First, we note that the term proportional to $u$ in $(\cD + u \cQ)\Phi(z)$ is zero. 
Using the first, second, and fifth identities we see that the term proportional to $u^0$ also vanishes. 
Finally, the term proportional to $u^{-1}$ is 
\[
2 \lambda_1^- \lambda_1^- \theta_3^+ - (\lambda_1^+\lambda_1^-) \theta_3^-  + (\lambda_1^- \lambda_3^+ - \lambda_1^+ \lambda_3^-) \theta_1^- . \]
The term proportional to $\theta_3^+$ vanishes by symmetry.
The terms proportional to $\theta_3^-, \theta_1^-$ vanish by the pure spinor constraint in Equation \eqref{eqn:6d10ps}. 

Next, we consider the anti-holomorphic coordinates. 
We record the following identities:
\begin{itemize}
\item $\cD (\zbar) = \lambda_3^+ \theta_1^- + \lambda_1^+ \theta_3^- - \lambda^-_3 \theta_1^+ - \lambda_1^- \theta_3^+$. 
\item $\cD(\theta_1^+ \theta_3^-) = - \lambda_1^+ \theta_3^- - \lambda_3^- \theta_1^+$. 
\item $\cD(\theta_1^- \theta_3^+) = \lambda_1^- \theta_3^+ + \lambda_3^+ \theta_1^-$. 
\item $\cQ (\zbar) = \theta_1^+$. 
\item $\cQ (\theta_1^+\theta_3^- + \theta_1^- \theta_3^+) = - \theta_1^+$. 
\end{itemize}
Using these identities we compute the total differential acting on $\Phi (\zbar)$
\[
(\cD + u \cQ) \Phi (\zbar) = 2 (\lambda_3^+ \theta_1^- - \lambda^-_3 \theta_1^+) .
\]
The right hand side is precisely $\Phi (u \d \zbar)$ as desired. 

By the same reasoning as in the previous case, we see that $\Phi$ preserves filtrations. 

\end{proof}

\begin{proof}[Proof for $d=10$]
We record the following identities.
\begin{itemize}
\item $\cD (z) = - \lambda_3 \theta_1 - \lambda_1 \theta_3$. 
\item $\cD (\theta_1 \theta_3) = \lambda_1 \theta_3 - \lambda_3 \theta_1.$
\item $\cD (\lambda_1 \theta_3 \theta_5) = \lambda_1 \lambda_3 \theta_5 - \lambda_1 \lambda_5 \theta_3$. 
\item $\cD (\lambda_3 \theta_3^2) = 2 \lambda_3^2 \theta_3$. 
\item $\cQ(\lambda_1 \theta_3 \theta_5) = - \lambda_1\theta_3$. 
\end{itemize}

We compute the total differential acting on 
\[
\Phi (z) = z - \theta_1 \theta_3 - 2 u^{-1} \lambda_1 \theta_3 \theta_5 +u^{-1} \lambda_3 \theta_3^2.
\]
First, we note that the term proportional to $u$ in $(\cD + u \cQ)\Phi(z)$ is zero. 
Using the first, second, and fifth identities we see that the term proportional to $u^0$ also vanishes. 
Using third and fourth identities the term proportional to $u^{-1}$ is
\[
- 2 \lambda_1 \lambda_3 \theta_5 + 2\lambda_1 \lambda_5 \theta_3 + \lambda_3^2 \theta_3 .
\]
This term vanishes by applying the pure spinor constraint,
Each term above vanishes by the pure spinor constraint in Equation \eqref{eqn:10d10ps}.

Next, we consider the anti-holomorphic coordinates.
From the identities 
\begin{itemize}
\item $\cD(\zbar) = - \lambda_5 \theta_1 - \lambda_3 \theta_3 - \lambda_1 \theta_5$,
\item $\cD (\theta_1 \theta_5) = \lambda_1 \theta_5 - \lambda_5 \theta_1$,
\end{itemize}
we read off
\[
(\cD + u \cQ) \Phi(\zbar) = - \lambda_3 \theta_3 - 2 \lambda_5 \theta_1 .
\]
The right hand side is precisely $\Phi (u \d \zbar)$ as desired. 

By the same reasoning as in the previous cases, we see that $\Phi$ preserves filtrations. 
\end{proof}

Consider the spectral sequence corresponding to the filtration on the localized twisted tautological bicomplex $(u)^{-1} A_Q$. 
Since $\Phi$ is a map of filtered cdgsa's, it defines a map of corresponding spectral sequences. 
Notice that the $E_1$ differential is identically zero. 
Already at this page we see that $\Phi$ induces an isomorphism. 

\begin{alem}
The map $\Phi$ induces an isomorphism on $E_1$-pages 
\[
  \Phi \colon  \Gr \left( (u)^{-1} \Dol(\CC^n) \right) \xto{\cong} H^\bu \left( \Gr \left(  (u)^{-1} A_Q \right) \right) .
\]
\end{alem}

To prove this lemma we will make use of an auxiliary $\ZZ$-grading on $\Gr A_Q$. 
Consider a new $\ZZ$-grading on $\Gr A_Q$ which assigns 
\begin{equation} \label{eqn:auxgr}
|\lambda|=|u|=|\theta_n| = 1
\end{equation}
with all other generators in degree zero.
Associated to this grading is an auxiliary filtration $\Tilde{F}^\bu \Gr A_Q$, which we will use to compute the cohomology.

As is clear from the discussion above, the associated graded of~$F^\bu A_Q$ is the tensor product of smooth functions on~$\R^d$ with a deformed version of the Koszul homology of~$Y$:
\deq{
  \Gr \left( (u)^{-1} A_Q \right) \cong C^\infty(\R^d) \otimes \left( K^\bu(Y)[u,u^{-1}], \d_K + u \pdv{ }{\theta_n} \right).
}
We will denote this deformed Koszul complex by $K^\bu_Q (Y)$. 
It is clear that functions on spacetime play no role at the $E_1$ page, so that we can think of the auxiliary filtration $\Tilde F^\bu$ as defined on~$K^\bu_Q (Y)$.

\begin{proof}[Proof for $d=4$]
Specializing to the case at hand, we can write the complex computing deformed Koszul homology explicitly in coordinates as follows:
\begin{equation}
  K^\bu_Q (Y) = \left(\frac{\C[\lambda_0, \lambda_1, \lambda_2]}{I} \otimes \C[\theta_0, \theta_1, \theta_2] [u, u^{-1}], ~ \sum
\lambda \pdv{ }{\theta} + u \pdv{ }{\theta_2}  \right) .
\end{equation}
Here, $I$ is the defining equation of the nilpotence variety, given by the equations
\deq{
  I = \langle \lambda_0 \lambda_1, \lambda_2 \lambda_1 \rangle.
}

Now consider the auxiliary filtration $\Tilde F^\bu K^\bu_Q$. 
This induces a spectral sequence $\{\Tilde E_r\}$ abutting to the deformed Koszul homology,
whose first page is
\deq{
  \Tilde E_1 = H^\bu\left(K_Q, (\lambda_2 + u) \pdv{ }{\theta_2} \right).
}
Taking cohomology of this differential has the effect of setting $\lambda_2 = -u$, which is equivalent to inverting $\lambda_2$. 
Geometrically, this corresponds to restricting to the complement of the hyperplane section $\lambda_2 = 0$ of~$Y$. The argument will proceed by arguing that the complement of this hyperplane section is in fact the product of $\C^\times \cong \Spec \C[\lambda_2,\lambda_2^{-1}]$ with an affine space.

Since $\lambda_2$ is now invertible, it follows from the second equation in~$I$ that $\lambda_1 = 0$; the first equation is then automatically satisfied. We therefore arrive at the description of the $\Tilde{E}_1$ page
\deq{
  \left( \Tilde E_1, \Tilde d_1 \right) = \left(  \C[\lambda_0] \otimes \C[u, u^{-1}] [\theta_0, \theta_1] , ~ \lambda_0 \pdv{ }{\theta_0} \right).
}
It is trivial to see that the spectral sequence $\{\Tilde E_r\}$ thus abuts to the graded commutative algebra
\deq{
 H^\bu(K_Q (Y)) = \C[u,u^{-1}, \theta_1]
}
The cohomology of $\Gr \left( (u)^{-1} A_Q \right)$ is thus the tensor product of this algebra with smooth functions on~$\R^4$. 
(We remark that the $E_2$ differential in the original spectral sequence is given by the only remaining component of $\cD_1$, which here is $2 u \theta_1 \cdot {\partial}/{\partial \zbar}$.)
The map $\Phi$ is the isomorphism that maps $f(z,\zbar) \d \zbar$ to $f(z,\zbar) \cdot 2\theta_1$ as desired. 
  \end{proof}
  
  \begin{proof}[Proof for $d=6$]
    In an explicit coordinate representation, the deformed Koszul cohomology is given by
    \deq{
      K_Q^\bu (Y) = \left( \frac{\C[\lambda_1^\pm,\lambda_3^\pm]}{I} \otimes \C[\theta_1^\pm,\theta_3^\pm][u, u^{-1}] \; , ~\; \sum \lambda \pdv{ }{\theta} + u \pdv{ }{\theta_3^-} \right),
    }
with $I$ defined by the equations 
    \deq[eq:6d-ideal]{
      I = \langle \lambda_1^+ \lambda_3^- - \lambda_1^- \lambda_3^+, \lambda_1^+ \wedge \lambda_1^-\rangle.
    }
    We consider the filtration $\Tilde F^\bu K_Q^\bu$. Just as above, the first page of the corresponding spectral sequence $\tilde E$ simply eliminates $\theta_3^-$ and sets $\lambda_3^-$ to $-u$, restricting us to the complement of the hyperplane section $\lambda_3^- = 0$. 

    As above, the complement of this hyperplane section is in fact affine. After inverting $\lambda_3^-$, the first defining equation in~\eqref{eq:6d-ideal} can be rewritten as
    \deq{
      \lambda_1^+ = - u^{-1} \lambda_3^+ \lambda_1^-.
    }
    The second equation is then automatically satisfied, due to the fact that the symplectic contraction $\lambda_1^- \wedge \lambda_1^-$ is identically zero. We thus see that 
$\left( \Tilde E_1, \Tilde d_1 \right)$ is  
\deq{\left( \C[u,u^{-1}] \otimes \C[\lambda_3^+, \lambda_1^-] \otimes \C[\theta_3^+, \theta_1^-,\theta_1^+], \lambda_3^+ \pdv{ }{\theta_3^+} + \lambda_1^- \pdv{ }{\theta_1^-} - u^{-1} \lambda_3^+ \lambda_1^- \pdv{ }{\theta_1^+} \right).
    }
    From this point, it is straightforward to argue (for example by first taking the cohomology of $\lambda_3^- \partial/\partial\theta_3^+$) that the spectral sequence abuts to 
    \deq{
      H^\bu(K_Q (Y)) = \C[u,u^{-1}][\theta_1^+].
    }
    The cohomology of $\Gr \left( (u)^{-1} A_Q \right)$ is obtained by tensoring with~$C^\infty(\R^6)$. The map $\Phi$ clearly induces the isomorphism $\d\zbar \mapsto 2 \theta_1^+$. 
  \end{proof}
  \begin{proof}[Proof for $d=10$]
  Specializing to the case at hand, we can write the complex computing deformed Koszul homology explicitly in coordinates as follows:
\begin{equation}
  K^\bu_Q (Y)  = \left(\frac{\C[\lambda_1, \lambda_3, \lambda_5]}{I} \otimes \C[\theta_1, \theta_3, \theta_5] [u, u^{-1}], ~ 
\sum \lambda \pdv{ }{\theta} + u \pdv{ }{\theta_5}  \right) .
\end{equation}
Here, $I$ is the defining equation of the nilpotence variety, given by the equations
    \deq{
      I = \left\langle \lambda_1 \lambda_5 + \frac{1}{2} \lambda_3 \lambda_3 , \lambda_3 \lambda_1 \right\rangle .
    }
Now consider the auxiliary filtration $\Tilde F^\bu K^\bu_Q$. 
This induces a spectral sequence $\{\Tilde E_r\}$ abutting to the deformed Koszul homology,
whose first page is
\deq{
  \Tilde E_1 = H^\bu\left(K_Q (Y) , (\lambda_5 + u) \pdv{ }{\theta_5} \right).
}
Taking cohomology of this differential has the effect of setting $\lambda_2 = -u$, which is equivalent to inverting $\lambda_2$. 
Geometrically, this corresponds to restricting to the complement of the hyperplane section $\lambda_5 = 0$ of~$Y$ just as in the previous cases. 
  
Since $\lambda_5$ is now invertible, it follows from the first equation in~$I$ that
\[
\lambda_1 = \frac12 u^{-1} \lambda_3 \lambda_3
\]
The second equation is satisfied tautologically, just for type reasons (there is no intertwiner of $\SU(5)$ representations of appropriate symmetry, since the symmetric cube of the three-form contains no four-form). 

We therefore arrive at the description of the $\Tilde{E}_1$ page
\deq{
  \left( \Tilde E_1, \Tilde \d_1 \right) = \left(  \C[\lambda_3] \otimes \C[u, u^{-1}] [\theta_1, \theta_3] , ~ \lambda_3 \frac{\partial}{\partial \theta_3} \right).
}
It is trivial to see that the spectral sequence $\{\Tilde E_r\}$ thus abuts to the graded commutative algebra
\deq{
H^\bu(K_Q (Y)) = \C[u,u^{-1}, \theta_1]
}
The cohomology of $\Gr \left( (u)^{-1} A_Q \right)$ is thus the tensor product of this algebra with smooth functions on~$\R^4$. 
(We remark that the $E_2$ differential in the original spectral sequence is given by the only remaining component of $\cD_1$, which here is $2 u \theta_1 \cdot {\partial}/{\partial \zbar}$.)
The map $\Phi$ is the isomorphism that maps $f(z,\zbar) \d \zbar$ to $f(z,\zbar) \cdot 2\theta_1$ as desired.   
  
\end{proof}

\section{Six-dimensional $\cN = (2,0)$ supersymmetry}\label{sec:6d}
\subsection{Details on six-dimensional $\N=(2,0)$ supersymmetry}
\label{ssec:2,0details}

The six-dimensional $(\N,0)$ supertranslation algebra is
\deq{
  (S_+ \otimes R_n)[-1] \oplus V[-2],
}
where $S_+$ is a irreducible semi-spin representation and $V \cong \CC^6$ the (complex) vector representation of $\so(6,\CC)$,
and $R_n = (\C^{2n}, \omega)$ is a symplectic vector space. 
The case $\N=(1,0)$ was studied above; here, we will be interested exclusively in the case $\N=2$, and thus set $R = R_2$. The $R$-symmetry in this case is $\lie{sp}(2)$.

As reviewed above, there is an isomorphism $\wedge^2 S_+ \cong V$ of~$\lie{so}(6)$ modules. The bracket is given by the tensor product of this isomorphism with the symplectic form, viewed as an $\lie{sp}(2)$-equivariant map $\wedge^2 R \rightarrow \C$.

Just as in the $\N=(1,0)$ case, a holomorphic supercharge $Q \in Y$ is a rank-one element of~$S_+ \otimes R$, and specifies both a maximal isotropic subspace $L \subset V$ and a line $\rho \subset R$. 
As a $\lie{gl}(3) \times \lie{sp}(1) \times \lie{gl}(1)$-module, the odd part of the algebra $S_+ \otimes R$ decomposes as
\deq{
  \wedge^\text{even} L^\vee \otimes \left( \C^- \oplus R^\circ \oplus \C_+ \right). 
}
with $R^\circ \cong (\CC^2, \omega^\circ) = R /\!/ \rho$ a two-dimensional symplectic vector space, isomorphic to the symplectic reduction of~$R$ along~$\rho = \C_+$.

We will denote the generators by their form degree, together with labels indicating the $R$-symmetry type: as such, a basis of the odd part of the superalgebra consists of
\deq{
Q_0^+, Q_0^\circ, Q_0^-; \quad Q_2^+, Q_2^\circ, Q_2^-.}
The holomorphic supercharge is $Q_0^+ \in \wedge^0 L^\vee \otimes \rho$. The nontrivial brackets in the algebra are 
\begin{equation} \label{eq:brackets20}
  \begin{gathered}
  [Q_0^+, Q_2^-] = \bar{P}; \quad [Q_0^-, Q_2^+] = \bar{P}; \quad [Q_0^\circ, Q_2^\circ] = \bar{P}; \\
  [Q_2^+, Q_2^-] = P ; \quad [Q_2^\circ, Q_2^\circ] = P.
\end{gathered}
\end{equation}
In the last commutator, a contraction with the symplectic form~$\omega^\circ$ is understood. 

From this computation, it is clear that the algebraic tangent space $T_{Q_0^+}Y$ consists of all supercharges other than $Q_2^-$, which (just as above) spans the normal bundle and witnesses the nullhomotopy of~$\bar{P}$. However, not all of these linear-order deformations integrate to finite deformations; the obstruction (responsible for the singularity of $Y$ at $Q_0^+$) is provided by the nontrivial self-bracket of $Q_2^\circ$. A deformation by such an element is unobstructed precisely when it is of rank one as an element of $\wedge^{2}L^\vee \otimes R^\circ$.

As above, we choose a coordinate $\lambda$ on the odd part of the supertranslation algebra, and write down explicit defining equations for the nilpotence variety. Using the brackets~\eqref{eq:brackets20}, it is straightforward to compute that these are given by
\begin{equation}\label{eqn:6d20ps}
  \lambda_3^-\lambda_1^+ + \lambda_3^+\lambda_1^- + \omega^\circ(\lambda_3^\circ, \lambda_1^\circ) = 0, \qquad 
  \lambda_1^-\lambda_1^+ + \frac{1}{2} \omega^\circ(\lambda_1^\circ, \lambda_1^\circ) = 0.
\end{equation}
We can also understand the structure of the dg Lie algebra $\lie{p}_Q$. 

\begin{prop}
As a cochain complex, the dg Lie algebra $\fp_Q$ is
\begin{equation}
  \begin{tikzcd}[row sep = 1 ex]
    0 & 1 & 2 \\ \hline \\
    \wedge^2 L^\vee \ar[r] & \wedge^2 L^\vee \otimes \rho \\
    \lie{sl}(L)  & \wedge^2 L^\vee \otimes \rho^\vee \ar[r] & L\\
    \C \ar[r] & \wedge^0 L^\vee \otimes \rho & L^\vee \\
    \lie{gl}(\rho) \ar[ru] \\
  (\rho^\vee)^{\otimes2} \ar[r] & \wedge^0 L^\vee \otimes \rho^\vee \\
  \rho^\vee \otimes R^\circ \ar[r] & \wedge^0 L^\vee \otimes R^\circ\\
  \lie{sp}(R^\circ)  & \wedge^2 L^\vee \otimes R^\circ \\
  \rho \otimes R^\circ \\
    \wedge^2 L \oplus \rho^{\otimes 2}.
  \end{tikzcd}
\end{equation}
The positively-graded piece of~$H^\bu(\lie{t}_Q)$ is isomorphic to 
\deq{
  (L \otimes R^\circ)[-1] \oplus \wedge^2 L[-2];
}
the bracket is the tensor product of wedging on~$L$ with the symplectic form $\omega^\circ$. 
The full $H^\bu(\ft_Q)$ is extended by the direct sum of the parabolic Lie algebras $(\lie{sl}(L) \oplus \wedge^2 L)$ and $(\lie{sp}(R^\circ) \oplus (\rho \otimes R^\circ) \oplus \rho^{\otimes 2})$ in degree zero, together with an additional $\lie{gl}(1)$ factor whose weights are determined by the kernel condition. 
\end{prop}

The nilpotence variety $Y_Q$ is therefore the space of rank-one two-by-three complex matrices, cut out by three quadratic equations in~$\C^{2\times 3} \cong L \otimes R^\circ.$

Let us spell out the $\cN=(2,0)$ tautological filtered cdga. 
As a graded commutative superalgebra, 
\deq{
  A = C^\infty (\CC^3)[\theta^\pm_1, \theta^\pm_3, \theta_1^\circ, \theta_3^\circ] \otimes \left(\frac{\CC[\lambda^\pm_1, \lambda^\pm_3, \lambda_1^\circ, \lambda_3^\circ]}{I}\right) .
}
The differential $\cD$ has the form
\begin{multline}
  \cD   =  - \lambda^+_3 \left(\frac{\partial}{\partial \theta^+_3} - \theta^-_1 \frac{\partial}{\partial \zbar} \right) - \lambda^+_1 \left(\frac{\partial}{\partial \theta^+_1} - \theta^-_3 \frac{\partial}{\partial \zbar} - \theta^-_1 \frac{\partial}{\partial z} \right) 
    + \lambda^-_3 \left(\frac{\partial}{\partial \theta^-_3} - \theta^+_1 \frac{\partial}{\partial \zbar} \right) \\ 
   + \lambda^-_1 \left(\frac{\partial}{\partial \theta^-_1} - \theta^+_3 \frac{\partial}{\partial \zbar} - \theta^+_1 \frac{\partial}{\partial z} \right) 
    +  \lambda^\circ_3 \left(\frac{\partial}{\partial \theta^\circ_3} - \theta^\circ_1 \frac{\partial}{\partial \zbar} \right) 
   + \lambda^\circ_1 \left(\frac{\partial}{\partial \theta^\circ_1} - \theta^\circ_3 \frac{\partial}{\partial \zbar} - \theta^\circ_1 \frac{\partial}{\partial z} \right) .
 \end{multline}
The holomorphic supercharge $Q$ acts on $A$ by the operator
\deq{
\cQ = \frac{\partial}{\partial \theta_3^-} + \theta_1^+ \frac{\partial}{\partial \zbar} .
}
As above, we are mostly interested in the twisted tautological cdgsa $A_Q$.

\subsection{The holomorphic twist}

In \S\ref{sec:smoothpure}, we defined $\Dol(X)$ as a filtered cdgsa.
We can extend this to a local cdgsa structure on the {\em full} Hodge filtration in the following fashion: we define 
\deq{
  \deRham(X) = \left( \Omega^{\bu,\bu}(X)[u], \del + u \dbar \right).
}
To view this as a cdgsa, the assignment of degrees is given by
\[
  |\d z_i| = (1, +) , \quad | \d \zbar_i | = (0, -), \quad |u| = (1,-).
\]
We can furthermore view this as a filtered cdgsa by defining 
\deq{
  F^\ell \deRham(X) = \bigoplus_{2i + j + k \geq \ell} u^k \Omega^{i,j}(X).
}
Note that $\d z$ is placed in filtration degree \emph{two}, so that the full differential $\del + u\dbar$ appears on the $E_2$ page of the  associated spectral sequence.

In similar fashion, we can define $\Dol^{\leq i}(X)$ to be the filtered cdgsa arising as the quotient of~$\deRham(X)$ by the ideal
\[
  \left(\Omega^{>i, \bu}(X) [u] [-i-1] , \partial + u \dbar\right) \subset \deRham(X).
\]

In the context of this section, we will be interested in the case $X = \C^3$ and $i=1$. 
We introduce the $\Dol^{\leq 1}(\C^3)$-module  
given by $\Dol(\C^3) \otimes \Pi R^\circ[-1]$, 
where $R^\circ$ is a symplectic vector space which we view as carrying filtration degree four. 
Here, elements $u^k f_I(z,\zbar) \d \zbar_I$ act by wedge product and elements $u^k f_{I,j}(z,\zbar) \d \zbar_I \d z_j$ act trivially. 
We can form the square-zero extension 
\deq{
B =   \Dol^{\leq 1}(\C^3) \ltimes \left( \Dol(\C^3)\otimes \Pi R^\circ[-1]  \right) 
}
of $\Dol^{\leq 1}(\C^3)$ by this module; this filtered cdgsa will play an important role in what follows.

\begin{athm}\label{thm:a6d}
There is a $\GL(3 ; \CC)$-equivariant quasi-isomorphism 
\deq{
  \Phi: (u)^{-1} B \rightarrow (u)^{-1} A_Q
}
of filtered local cdgsa's on $\CC^{3}$.
\end{athm}

The map $\Phi: (u)^{-1} B \to (u)^{-1} A_Q$ is defined on generators as follows:
\beqn
\begin{aligned}
 z & \mapsto  z + \theta^-_1 \theta^+_1 + 2 u^{-1} \theta^- (\lambda_1^- \theta_3^+ - \lambda_1^+ \theta_3^- + \lambda_1^\circ \theta_3^\circ) + u^{-1} \theta_1^\circ (\lambda_1^\circ \theta_3^- + \lambda_1^- \theta_3^\circ)   \\
 \zbar &\mapsto \zbar + \theta^+_1 \theta^-_3 + \theta^-_1 \theta^+_3 + \theta_1^\circ \theta_3^\circ \\
 \d z & \mapsto  - u^{-1} (\lambda_3^- \lambda_1^\circ + \lambda_1^- \lambda_3^\circ) \theta_1^\circ \\
 \d \zbar & \mapsto  2 u^{-1} (\lambda_3^+ \theta_1^- - \lambda^-_3 \theta_1^+ + \lambda_3^\circ \theta_1^\circ) \\
 1 \otimes r^\circ & \mapsto  - u^{-1} (\theta_3^- \lambda_1^\circ + \lambda_1^- \theta_3^\circ) (\lambda_3^- \lambda_1^\circ + \lambda_1^- \lambda_3^\circ) \theta_1^\circ .
\end{aligned}
\eeqn

\begin{alem}
$\Phi$ is a map of filtered cdgsa's. 
\label{lemmaB:6d2,0}
\end{alem}

\begin{proof}
We need to show that $\Phi$ is a map of cochain complexes and that it is compatible with the algebra structure. 
We begin by showing that $\Phi$ is compatible with the differentials.

Consider the holomorphic coordinate $z$. 
Write $\Phi (z) = \phi_0 (z) + u^{-1} \phi_{-1} (z)$ where 
\begin{align*}
\phi_0(z) &= z + \theta_1^- \theta_1^+ \\
\phi_{-1} (z) & = 2 \theta^- (\lambda_1^- \theta_3^+ - \lambda_1^+ \theta_3^- + \lambda_1^\circ \theta_3^\circ) + \theta_1^\circ (\lambda_1^\circ \theta_3^- + \lambda_1^- \theta_3^\circ) .
\end{align*}
Notice that the total differential applied to $\Phi(z)$ has the following form
\begin{equation}\label{6d20a}
(\cD + u \cQ) \Phi(z) = u \cQ \phi_0(z) + (\cD \phi_0(z) + \cQ \phi_{-1}(z)) + u^{-1} \cD \phi_{-1} (z) .
\end{equation}
Clearly the term proportional to $u$ vanishes. 
Next, notice that
\begin{align*}
\cD \phi_0 (z) & = 2 \lambda_1^+ \theta_1^- + \lambda^\circ \theta_1^\circ \\
\cQ \phi_{-1} (z) & = - 2 \lambda_1^+ \theta_1^- - \lambda^\circ \theta_1^\circ
\end{align*}
so that the term in \eqref{6d20a} proportional to $u^0$ vanishes. 

Finally the term in \eqref{6d20a} proportional to $u^{-1}$ is
\begin{align*}
\cD \phi_{-1}(z) & = 2 \lambda_1^- \lambda_1^- \theta_3^+ - (\lambda_1^- \lambda_1^+ \pm \omega^\circ (\lambda_1^\circ, \lambda_1^\circ)) \theta_3^- \\ & + (\lambda_1^+ \lambda_3^- - \lambda_1^- \lambda_3^+ + \lambda_1^\circ \lambda_3^\circ) \theta_1^- + (\lambda_1^- \lambda_1^\circ \pm \lambda_1^- \lambda_1^\circ) \theta_3^\circ \\ & - (\lambda_3^- \lambda_1^\circ + \lambda_1^- \lambda_3^\circ) \theta_1^\circ .
\end{align*}
Clearly the term proportional to $\theta_3^\circ$ vanishes; the term proportional to~$\theta_3^+$ vanishes by symmetry. 
The terms proportional to $\theta_3^-, \theta_1^-$ both vanish by the pure spinor constraints~\eqref{eqn:6d20ps}.
The remaining term is $\Phi ( \d z)$, as desired.

Next, we consider the anti-holomorphic coordinate $\zbar$. 
We record the following identities:
\begin{itemize}
\item $\cD (\zbar) = \lambda_3^+ \theta_1^- + \lambda_1^+ \theta_3^- - \lambda^-_3 \theta_1^+ - \lambda_1^- \theta_3^+ + \lambda_3^\circ \theta_1^\circ  + \lambda_1^\circ \theta_3^\circ$. 
\item $\cD(\theta_1^+ \theta_3^-) = - \lambda_1^+ \theta_3^- - \lambda_3^- \theta_1^+$. 
\item $\cD(\theta_1^- \theta_3^+) = \lambda_1^- \theta_3^+ + \lambda_3^+ \theta_1^-$. 
\item $\cD (\theta_1^\circ \theta_3^\circ) =  \lambda_3^\circ \theta_1^\circ - \lambda_1^\circ \theta_3^\circ$. 
\item $\cQ (\zbar) = \theta_1^+$. 
\item $\cQ (\theta_1^+\theta_3^- + \theta_1^- \theta_3^+ + \theta_1^\circ \theta_3^\circ) = - \theta_1^+$. 
\end{itemize}

Using these identities we compute the total differential acting on $\Phi (\zbar) = \zbar + \theta_1^+ \theta_3^- + \theta_1^-\theta_3^+ - \theta_1^\circ \theta_3^\circ$. 
\[
  \left( \cD + u \cQ \right)  \Phi (\zbar) = 2 (\lambda_3^+ \theta_1^- - \lambda^-_3 \theta_1^+ + \lambda_3^\circ \theta_1^\circ) .
\]
The right-hand side is precisely $\Phi (u \d \zbar)$ as desired. 

Finally, we consider the element $r^\circ \in B$.
Notice that
\[
\Phi (r^\circ) = (\theta_3^- \lambda_1^\circ + \lambda_1^- \theta_3^\circ) \Phi (\d z) .
\]
Thus
\[
(\cD + u \cQ) \Phi(r^\circ) = (\lambda_3^- \lambda_1^\circ + \lambda_1^- \lambda_3^\circ) \Phi(\d z) .
\] 
The following lemma implies $(\cD + u \cQ) \Phi(r^\circ) = 0$.

\begin{lem}\label{lem:6d20a}
$(\lambda_3^- \lambda_1^\circ + \lambda_1^- \lambda_3^\circ)^2 = 0 .$
\end{lem}

\begin{proof}
This is a direct calculation. 
One has
\begin{align*}
(\lambda_3^- \lambda_1^\circ + \lambda_1^- \lambda_3^\circ)^2 & = (\lambda_3^- \lambda_1^\circ)^2 + (\lambda_1^- \lambda_3^\circ)^2 + 2 \lambda_3^- \lambda_1^\circ \lambda_3^\circ \lambda_1^- \\ & = - (\lambda_3^-)^2 \lambda_1^+ \lambda_1^- + \lambda_3^- \lambda_1^- (\lambda_1^+ \lambda_3^- + \lambda_1^- \lambda_3^+) \\ & = 0
\end{align*}
The second line follows from the pure spinor constraints \eqref{eqn:6d20ps}. 
The third line follows from the fact that $\lambda_1^- \lambda_1^- = 0$ by symmetry.
\end{proof}

Finally, we show that $\Phi$ defines a map of algebras. 
It suffices to verify the following relations:
\begin{itemize}
\item[(a)] $\Phi (\d z)^2 = 0$. 
\item[(b)] $\Phi( \d z) \Phi (r^\circ) = 0$. 
\item[(c)] $\Phi(r^\circ)^2 = 0$. 
\end{itemize}

Each of these relations hold by applying Lemma \ref{lem:6d20a}. 
Indeed, relation (a) is equivalent to
\[
 (\lambda_3^- \lambda_1^\circ + \lambda_1^- \lambda_3^\circ)^2 (\theta_1^\circ)^2  = 0 .
\]
Relation (b) is equivalent to 
\[
 (\theta_3^- \lambda_1^\circ + \lambda_1^- \theta_3^\circ)  (\lambda_3^- \lambda_1^\circ + \lambda_1^- \lambda_3^\circ)^2  = 0.
\]
Finally, relation (c) is equivalent to 
\[
(\theta_3^- \lambda_1^\circ + \lambda_1^- \theta_3^\circ)^2  (\lambda_3^- \lambda_1^\circ + \lambda_1^- \lambda_3^\circ)^2 = 0 .
\]
This completes the proof.
\end{proof}

\begin{alem}
  The map $\Phi$ induces an isomorphism of $E_1$ pages
  \deq{
    \Phi: \Gr \left( (u)^{-1} B \right)  \xto{\cong} H^\bu \left( \Gr \left( (u)^{-1} A_Q \right) \right).
  }
  \begin{proof}
    The structure is completely analogous to the proof of the smooth Lemma~C above, with the new ingredient that the deformed Koszul homology is now more interesting. 
    With the help of the description reviewed above in~\S\ref{ssec:2,0details}, we can describe this complex explicitly:
    \deq{
      %(K, d_K)
      K_Q^\bu (Y) = \left( \frac{\C[\lambda_1^\pm, \lambda_1^\circ,\lambda_3^\pm,\lambda_3^\circ]}{I}\otimes  \C[\theta_1^\pm, \theta_1^\circ,\theta_3^\pm,\theta_3^\circ][u,u^{-1}], ~  \sum \lambda \pdv{ }{\theta}  + u \pdv{ }{\theta_3^-} \right).
    }
    Here, the defining equations of the ideal $I$ are 
    \deq[eq:2,0ideal]{
      I = \left\langle   \lambda_3^- \lambda_1^+ + \lambda_3^+ \lambda_1^- + \omega(\lambda_3^\circ, \lambda_1^\circ) , 
      \lambda_1^- \lambda_1^+ + \frac{1}{2}\omega^\circ(\lambda_1^\circ, \lambda_1^\circ) \right\rangle .
    }
    We consider the auxiliary filtration $\Tilde F^\bu K_Q^\bu$.
    Passing to the first page of the induced spectral sequence again amounts to setting $\lambda_3^- = - u$, and thus to restricting to the complement of the hyperplane section $\lambda_3^- = 0$. Having done this, we can solve for $\lambda_1^+$ using the first of the defining equations in~\eqref{eq:2,0ideal}: 
    \deq[eq:2,0sol]{
      \lambda_1^+ = u^{-1} \left( \lambda_3^+ \lambda_1^- + \omega^\circ(\lambda_3^\circ, \lambda_1^\circ)  \right).
    }
    The structure of the computation now exhibits a new feature, compared with the examples in the previous section: Because the holomorphic supercharge no longer defines a smooth point of~$Y$, the second set of defining equations are not tautologically satisfied. Indeed, substituting~\eqref{eq:2,0sol} into the second equation produces
    \deq{
      \lambda_1^- u^{-1} \left( \lambda_3^+ \lambda_1^- + \omega^\circ(\lambda_3^\circ, \lambda_1^\circ) \right)  + \frac{1}{2} \omega^\circ(\lambda_1^\circ, \lambda_1^\circ) = 0.
    }
    The first of these terms vanishes for type reasons, since the required contraction of Lorentz representations is antisymmetric. If we define
    \deq[eq:lambdatilde]{
      \tilde \lambda_1^\circ = \lambda_1^\circ + u^{-1} \lambda_3^\circ \lambda_1^{-},
    }
    the remaining equations can be written in the form 
    \deq{
      \omega^\circ(\tilde\lambda_1^\circ, \tilde\lambda_1^\circ) = 0.
    }
    These are the defining equations for the Segre variety of two-by-three matrices of rank one, which (as recalled above) is the nilpotence variety $Y_Q$ of $H^\bu(\lie{p}_Q)$. Note that the quadratic term in $\lambda_3^\circ \lambda_1^{-}$ vanishes, since this expression is automatically of rank one. 
We will denote the ideal generated by these equations by~$J$.
    In light of this computation, we have 
    \deq{
      \Tilde E_1 =  \frac{\C[\tilde\lambda_1^\circ]}{J}\otimes \C [\lambda_1^-,\lambda_3^\circ,\lambda_3^+][\theta_1^\pm,\theta_1^\circ,\theta_3^+,\theta_3^\circ][u,u^{-1}],
    }
    equipped with the differential
\begin{multline}
      \Tilde \d_1 = \lambda_1^- \pdv{ }{\theta_1^-} + \lambda_3^\circ \pdv{ }{\theta_3^\circ} + \lambda_3^+ \pdv{ }{\theta_3^+} \\ + u^{-1} \left( \lambda_3^+ \lambda_1^- + \omega(\lambda_3^\circ, \lambda_1^\circ)  \right) \pdv{ }{\theta_1^+} + \left( \lambda_1^\circ - u^{-1} \lambda_3^\circ \lambda_1^{-} \right) \pdv{ }{\theta_1^\circ}.
\end{multline}
Contracting the acyclic pieces of the differential, it follows immediately that the spectral sequence abuts to 
\deq{
  H^\bu(K_Q(Y) ) \cong \C[u,u^{-1}][\theta_1^+] \otimes_\C H^\bu(K(Y_Q)).
 }
\begin{rmk}
The complex $K_Q^\bu(Y)$ should {\em not} be confused with the complex $K^\bu (Y_Q)$. 
The first is obtained by deforming the original Koszul complex of $Y$ by the differential $u \partial_{\theta_3^-}$.
The latter is the Koszul complex of the nilpotence variety $Y_Q$ associated to the graded Lie algebra $\fp_Q$. 
\end{rmk}

 The cohomology of~$\Gr\left( (u)^{-1} A_Q \right)$ is obtained by taking the tensor product with smooth functions on~$\R^6$. 
  \end{proof}
\end{alem}

\begin{rmk}
We note that the proposition immediately implies that 
\deq{
  \Gr \left( (u)^{-1} B \right) \cong \Gr \left( (u)^{-1} \Dol(\R^6) \right) \otimes_\C H^\bu(K(Y_Q)).
}
In fact, because it implies an isomorphism of spectral sequences, it establishes the stronger result that
\deq{
  (u)^{-1} B \cong (u)^{-1} \Dol(\R^6) \otimes_\C H^\bu(K(Y_Q)),
}
where the differential on the right-hand side contains additional terms (corresponding to the $\del$ differential in~$B$) that arise from the nontrivial brackets in the algebra $H^\bu(\lie{p}_Q)$. 
This shows that we can think of $B$ as being the tautological filtered cdga associated to~$H^\bu(\lie{p}_Q)$. 
\end{rmk}

\section{Eleven-dimensional supersymmetry}
\label{sec:11d}
\subsection{Details on eleven-dimensional supersymmetry}
The eleven-dimensional supertranslation algebra is
\deq{
  \lie{t} = S[-1] \oplus V[-2],
}
where $S$ is the (unique) spin representation and $V \cong \C^{11}$ the complex vector representation, of~$\lie{so}(11, \CC)$. 
The bracket is the unique surjective $\lie{so}(11,\CC)$-equivariant map from the symmetric square of~$S$ to~$V$;
this decomposes into three irreducibles, 
\deq{
  \Sym^2(S) \cong V \oplus \wedge^2 V \oplus \wedge^5 V.
}
The $R$-symmetry is trivial.

Choosing a maximal isotropic subspace $L \subset V$, we can decompose the algebra into $\lie{sl}(5)$ representations, using the isomorphisms
\deq{
  V = L \oplus L^\vee \oplus \C, \qquad S = \wedge^\bu L^\vee.
}
(Recall from above that the copy of $\C$ in the decomposition of~$V$ can be identified with~$L^\perp/L$.) Recalling the construction of Clifford multiplication, we can then write the brackets in the form
\begin{equation}
  \begin{aligned}
  [Q_0, Q_4] = [Q_1, Q_3] = [Q_2, Q_2] &= \bar{P}; \\
  [Q_0, Q_5] = [Q_1, Q_4] = [Q_2, Q_3] &= H; \\
  [Q_1, Q_5] = [Q_2, Q_4] = [Q_3, Q_3] &= P.
\end{aligned}
\end{equation}
From here, it follows immediately that the nilpotence variety is cut out by the following equations
\begin{align*}
\lambda_0 \lambda_4 + \lambda_1 \lambda_3 + \frac{1}{2} \lambda_2^2 & = 0,\\
\lambda_0 \lambda_5 + \lambda_1 \lambda_4 + \lambda_2 \lambda_3 & = 0,\\
\lambda_1 \lambda_5 + \lambda_2 \lambda_4 + \frac{1}{2} \lambda_3^2 & = 0.
\end{align*}
Denote by $I$ the idea generated by these equations.

This space was used in the eleven-diomensional pure spinor superfield formalism~\cite{Ced-towards,Ced-11d}; see also~\cite{character}. A new feature in eleven dimensions is that---despite the absence of $R$-symmetry---the nilpotence variety is nevertheless singular and contains strata corresponding to different possible twists. From the above decomposition, it is clear that any holomorphic supercharge $Q \in  \wedge^0 L^\vee$ also pairs with itself to zero in $\wedge ^2 V$. Other supercharges exist for which this two-form pairing does not vanish, even though $[Q,Q] = 0$. Such supercharges are well-defined on manifolds of $SU(2) \times G_2$ holonomy; for details, we refer to~\cite{Awada,Duff,Joyce}. 

   The structure of the singularities is easiest to understand, as in other cases, by considering the dg Lie algebra $\lie{p}_Q$. 
   
\begin{prop}
As a cochain complex, the dg Lie algebra $\fp_Q$ is
\[
    \begin{tikzcd}[row sep = 1 ex]
    0 & 1 & 2 \\ \hline \\ 
    L \otimes \C & \wedge^5 L^\vee \ar[r] & \C \\
    \wedge^2 L & \wedge^4 L^\vee \ar[r] & L\\
    \lie{sl}(L)  & \wedge^3 L^\vee & L^\vee \\
    \wedge^2 L^\vee \ar[r] & \wedge^2 L^\vee \\
    L^\vee \otimes \C \ar[r] & \wedge^1 L^\vee \\
    \lie{gl}(1) \ar[r] & \wedge^0 L^\vee.
  \end{tikzcd}
\]
The positively graded piece of~$H^\bu(\lie{p}_Q)$ is $\wedge^2 L[-1] \oplus \wedge^4 L[-2]$, with bracket just defined by the wedge product. 
The full $H^\bu(\fp_Q)$ is extended by the parabolic Lie algebra $\lie{sl}(L) \oplus L \oplus \wedge^2 L$ in degree zero.
\end{prop}

Let us spell out the tautological filtered cdgsa using these coordinates. 
As a commutative graded super algebra
\deq{
  A = \rC^\infty (\CC^5 \times \RR)[\theta_0,\ldots, \theta_5] \otimes \left(\frac{\CC[\lambda_0, \ldots,\lambda_5]}{I}\right) .
}
The explicit coordinate form of the differential is $\cD = \sum \lambda_i \cD_{5-i}$ where 
\begin{align*}
\cD_0 & = \frac{\partial}{\partial \theta_5} - \theta_1 \frac{\partial}{\partial \zbar} - \theta_0\frac{\partial}{\partial t} \\
\cD_1 & = \frac{\partial}{\partial \theta_4} - \theta_2 \frac{\partial}{\partial \zbar} - \theta_1 \frac{\partial}{\partial t} - \theta_0 \frac{\partial}{\partial z} \\
\cD_2 & = \frac{\partial}{\partial \theta_3} - \theta_3 \frac{\partial}{\partial \zbar} - \theta_2 \frac{\partial}{\partial t} - \theta_1 \frac{\partial}{\partial z} \\
\cD_3 & = \frac{\partial}{\partial \theta_2} - \theta_4 \frac{\partial}{\partial \zbar} - \theta_3 \frac{\partial}{\partial t} - \theta_2 \frac{\partial}{\partial z} \\
\cD_4 & = \frac{\partial}{\partial \theta_1} - \theta_5 \frac{\partial}{\partial \zbar} - \theta_4 \frac{\partial}{\partial t} - \theta_3 \frac{\partial}{\partial z} \\
\cD_5 & = \frac{\partial}{\partial \theta_0} - \theta_5 \frac{\partial}{\partial t} - \theta_4 \frac{\partial}{\partial z} 
\end{align*}
The action of the element $Q$ on $A$ is through the operator
\deq{
  \cQ = \pdv{ }{\theta_5} + \theta_1 \pdv{}{\zbar} + \theta_0 \pdv{}{t} .
}

\subsection{The twist of the tautological cdga}
 
In this section, our strategy of proof is slightly different than above. We do not explicitly write a map $\Phi$ of filtered cdgsa's (although we believe it should be possible to construct such a map in this example as well). As such, Theorem~A and Lemma~B are omitted. Completing the characterization of the cohomology of~$A_Q$ thus requires an explicit check that the further differentials in fact agree, but this is straightforward. In earlier examples, the map $\Phi$ induced an isomorphism of spectral sequences, so that no further check was necessary.

As in the previous sections, let $Y_Q$ be the nilpotence variety associated to the graded Lie algebra $H^\bu (\fp_Q)$. 
The following result characterizes the deformed tautological complex $A_Q$ in terms of the Koszul cohomology of $Y_Q$.

\addtocounter{athm}{2}

\begin{alem}
\label{alem11d}
There is an isomorphism of commutative graded super algebras
\[
\Dol(\CC^5 \times \RR) \otimes H^\bu (K(Y_Q)) \cong H^\bu \left(\Gr \left((u)^{-1} A_Q\right) \right)  .
\]
\end{alem}

In the statement of the lemma, we have viewed $\Dol(\CC^5 \times \RR)$ as a cdgsa with {\em zero} differential.

  \begin{proof}
    As always, the $E_1$ page is the tensor product of smooth functions on spacetime with the deformed Koszul cohomology $H^\bu(K_Q)$. 
    In coordinates, $K_Q^\bu$ is given by
    \deq{
     K_Q^\bu (Y) =  \left(\frac{\CC[\lambda]}{I} \otimes \C[\theta][u,u^{-1}] \;,~\; \lambda \frac{\partial}{\partial \theta} + u \pdv{ }{\theta_5} \right).
    }
   The ideal $I$ is generated by the eleven equations
   \deq{
     I = \left\langle \lambda_0 \lambda_4 + \lambda_1 \lambda_3 + \frac{1}{2} \lambda_2^2,
     \lambda_0 \lambda_5 + \lambda_1 \lambda_4 + \lambda_2 \lambda_3,
     \lambda_1 \lambda_5 + \lambda_2 \lambda_4 + \frac{1}{2} \lambda_3^2 \right\rangle.
   }
We consider the auxiliary filtration $\Tilde{F}^\bu K_Q (Y)$. 
Passing to the first page of the induced spectral sequence, which abuts to $H^\bu(K_Q (Y))$, sets $\lambda_5 = - u$, thus effectively inverting $\lambda_5$ and restricting to the complement of the hyperplane section $\lambda_5  =  0$. 
Having done this, we can solve six of the above eleven equations, obtaining
   \deq{
     \lambda_0 = u^{-1} \left( \lambda_1 \lambda_4 + \lambda_2 \lambda_3 \right), \qquad
     \lambda_1 = u^{-1} \left( \lambda_2 \lambda_4 + \frac{1}{2} \lambda_3^2 \right).
   }
   Just as for $\N=(2,0)$ supersymmetry in six dimensions, the holomorphic supercharge does not sit at a singular point of~$Y$, so that the ideal $I$ is nontrivial even after inverting $\lambda_5$. After solving for $\lambda_0$ and~$\lambda_1$, the remaining equations become
   \deq{
     0 =      u^{-1} \left( u^{-1} \left( \lambda_2 \lambda_4 + \frac{1}{2} \lambda_3^2 \right) \lambda_4 + \lambda_2 \lambda_3 \right) \lambda_4 + u^{-1} \left( \lambda_2 \lambda_4 + \frac{1}{2} \lambda_3^2 \right) \lambda_3 + \frac{1}{2} \lambda_2^2.
}
For type  reasons, the terms containing $\lambda_3^3$ and $\lambda_4^3$ vanish, so that the equation reduces to
\deq{
  0 = \frac{1}{2} \lambda_2^2 + u^{-1} \lambda_2 \lambda_3 \lambda_4 + \frac{1}{2} u^{-2} \lambda_3^2 \lambda_4^2 
  =  \frac{1}{2} \left( \lambda_2 + u^{-1} \lambda_3 \lambda_4 \right) ^2.
}
Defining the new coordinate  $\tilde\lambda_2 =  \lambda_2 + u^{-1} \lambda_3 \lambda_4$, we can write the ideal $J$ generated by the remaining equation simply as $J = \langle \tilde\lambda_2^2 \rangle$. In light of this computation, we have
\deq{
  \Tilde{E}_1 = \frac{\C[\tilde\lambda_2]}{J} \otimes \CC[\lambda_3,\lambda_4] [\theta_0,\theta_1,\theta_2,\theta_3,\theta_4][u,u^{-1}].
}
Contracting the acyclic pieces of the differential, we immediately obtain that
\deq{
  H^\bu(K_Q) \cong \C[\theta_0,\theta_1] [u,u^{-1}] \otimes_\C H^\bu(K(Y_Q)),
}
where $K^\bu(Y_Q) = \C[\theta_2,\tilde\lambda_2]/J$ is the Koszul complex of the maximal ideal for the variety defined by the ideal $J$ . 
  \end{proof}

We point out that the ideal $J$ appearing at the end of the exhibits $Y_Q$ as the affine cone over the Grassmannian $\Gr(2,5)$.

\subsection{Cohomology of $Y_Q$}\label{ssec:koszul}

Let $z$ be a holomorphic coordinate on $\CC^5$ and $t$ a real coordinate on $\RR$. 
By Lemma \ref{alem11d} the cohomology of the associated graded of $(u)^{-1} A_Q$ is given by 
\beqn\label{eqn:11dsuperfield}
\rC^\infty(\CC^5 \times \RR) [\theta_0,\theta_1] \otimes H^\bu (K(Y_Q))[u,u^{-1}] 
\eeqn
where $\theta_0$ plays the role of the one-form $\d t$ and $\theta_1$ plays the role of the one-form $\d \zbar$. 
We first recall a description of the Koszul cohomology of $Y_Q$ (being the affine cone over ${\rm Gr}(2,5)$ this is a standard result). 
We remark that the cohomology is bigraded by $(\lambda,\theta)$ degree. 
Then, we describe the differential present at the next page of the spectral sequence abutting to the cohomology of $A_Q$.

Let $\alpha = -\theta_2 \Tilde{\lambda}_2 \in L$.
The following is a complete description of the cohomology of $K^\bu(Y_Q)$. 
\begin{itemize}
\item In degree zero the cohomology is one-dimensional spanned by ~$1$. 
\item In degree $(1,1)$ the cohomology is five-dimensional spanned by $\alpha$. 
\item In degree $(1,2)$ the cohomology is five-dimensional spanned by $\alpha \theta_2$. 
\item In degree $(2,3)$ the cohomology is one-dimensional spanned by $\alpha^2 \theta_2$. 
\end{itemize}

The remaining differential acting on the graded algebra \eqref{eqn:11dsuperfield} coming from the spectral sequence abutting to $A_Q$ is of the form
\beqn\label{eqn:11diff1}
\Tilde{\d} = - \Tilde{\lambda}_2 \theta_2 \frac{\partial}{\partial z} + u \theta_1 \frac{\partial}{\partial \zbar} + u \theta_0 \frac{\partial}{\partial t}.
\eeqn
From this, we can characterize the action of the differential on each of the classes above.
Let $f = f(z,\zbar, t , \theta_0, \theta_1)$ denote an element of $\rC^\infty(\CC^5 \times \RR)[\theta_0,\theta_1]$.
Then
\beqn\label{eqn:11diff2}
\Tilde{\d} (f) = (\partial_z f) \alpha + u (\partial_{\zbar} f) \theta_1 + u(\partial_t \alpha) \theta_0 .
\eeqn
The value of the differential on the other generators is determined by the Leibniz rule.

\section{Ten-dimensional $(2,0)$ supersymmetry} \label{sec:IIB}
\subsection{Details on ten-dimensional $\N=(2,0)$ supersymmetry}
The $\N$-extended chiral supersymmetry algebra in ten dimensions is constructed by choosing an $\N$-dimensional vector space $R$ with a symmetric bilinear pairing. The $R$-symmetry is then $\lie{so}(R)$. We will be interested in the case $\N=(2,0)$, relevant for type IIB supergravity; the supertranslation algebra then takes the form
\deq{
  \lie{t} = (S_+ \otimes R)[-1] \oplus V[-2].
}
Holomorphic supercharges are of rank one with respect to the tensor product decomposition; we can fix a holomorphic supercharge by specifying a maximal isotropic subspace $L\subset V$, together with a decomposition $R = \rho \oplus \rho^\vee$. The supersymmetry algebra then decomposes as 
\deq{
\left( \wedge^\text{even}(L^\vee) \otimes \left( \rho \oplus \rho^\vee \right) \right)[-1] \oplus \left( L \oplus L^\vee \right)[-2].
}
The brackets in the algebra then take the form
\deq{
  [Q_0^\pm, Q_4^\mp] = [Q_2^+, Q_2^-] = \bar{P}, \qquad
  [Q_2^\pm, Q_4^\mp] = P,
}
where the $\pm$ superscripts denote homogeneous decomposable elements in $\rho$ or~$\rho^\vee$ respectively.
From here, it is straightforward to see that the defining equations of the nilpotence variety are 
\deq[eq:I-IIB]{
  \lambda_1^- \lambda_5^+ + \lambda_1^+ \lambda_5^- + \lambda_3^+ \lambda_3^- = 0, \qquad
  \lambda_1^+ \lambda_3^- + \lambda_1^- \lambda_3^+ = 0.
}

The dg Lie algebra $\lie{p}_Q$ is described as follows,
see \cite[Lemma 6.5.1]{CLsugra}. 

\begin{prop}
As a cochain complex, the dg Lie algebra $\fp_Q$ is
  \[
    \begin{tikzcd}[row sep = 1 ex]
    0 & 1 & 2 \\ \hline \\ 
    \wedge^2 L & \wedge^4 L^\vee \otimes \rho^\vee \ar[r] & L\\
    \lie{sl}(L)  & \wedge^4 L^\vee \otimes \rho & L^\vee \\
    & \wedge^2 L^\vee \otimes \rho^\vee \\
  \wedge^2 L^\vee \ar[r] & \wedge^2 L^\vee\otimes \rho \\
    \lie{gl}(1) \ar[r] & \wedge^0 L^\vee \otimes \rho \\
    \lie{gl}(\rho) \ar[r] & \wedge^0 L^\vee \otimes \rho^\vee
  \end{tikzcd}
\]
The positively graded piece is isomorphic to
\deq{
  \left( L \oplus \wedge^2 L^\vee \right) [-1] \oplus L^\vee[-2],
}
where the bracket is just given by contraction. 
The full $H^\bu(\fp_Q)$ is extended by the same parabolic Lie algebra as in the $\N=(1,0)$ case, namely $\lie{sl}(L) \oplus \wedge^2 L$, sitting in degree zero.
\end{prop}

Let us spell out the tautological filtered cdgsa using these coordinates. 
As a commutative graded superalgebra,
\deq{
  A = \rC^\infty (\CC^5)[\theta^\pm] \otimes \left(\frac{\CC[\lambda^\pm ]}{I}\right) .
}
where $\theta^\pm = (\theta_1^\pm,\theta_3^\pm, \theta_5^\pm)$ and similarly for the $\lambda$ coordinates. 
The ideal $I$ is defined by the equations in~\eqref{eq:I-IIB}.
The explicit coordinate form of the differential is
\begin{multline}
  \cD = \lambda_5^+ \left(\pdv{ }{\theta_5^+} - \theta_1^- \pdv{ }{\zbar}\right) + \lambda_3^+ \left(\pdv{ }{\theta_3^+} - \theta_3^- \pdv{ }{\zbar} - \theta_1^- \pdv{ }{z}\right) + \lambda_1^+ \left(\pdv{ }{\theta_1^+} - \theta_5^- \pdv{ }{\zbar} - \theta_3^- \pdv{ }{z}\right) \\
+ \lambda_5^- \left(\pdv{ }{\theta_5^-} - \theta_1^+ \pdv{ }{\zbar}\right) + \lambda_3^- \left(\pdv{ }{\theta_3^-} - \theta_3^+ \pdv{ }{\zbar} - \theta_1^+ \pdv{ }{z}\right) + \lambda_1^- \left(\pdv{ }{\theta_1^-} - \theta_5^+ \pdv{ }{\zbar} - \theta_3^+ \pdv{ }{z}\right) .
\end{multline}
The action of the element $Q_0^+$ on $A$ is through the operator
\deq{
  \cQ = \pdv{ }{\theta_5^-} + \theta_1^+ \pdv{ }{\zbar} .
}

\subsection{The twist of the tautological cdga}

In this section our strategy is similar to Section \ref{sec:11d}. 
The following lemma gives a description of the cohomology of the twisted tautological complex $A_Q$. 

\addtocounter{athm}{2}

\begin{alem}
\label{alemIIB}
There is an isomorphism of commutative graded super algebras
\[
  \Dol(\CC^5) \otimes H^\bu (K(Y_Q)) \cong H^\bu \left(\Gr \left((u)^{-1} A_Q\right) \right).
\]
\end{alem}
In the statement of the lemma, we have viewed $\Dol(\CC^5 \times \RR)$ as a cdgsa with {\em zero} differential.
\begin{proof}
  As is by now familiar, we can ignore the  factor  of $C^\infty(\R^{10})$ and simply study the cohomology of~$K^\bu_Q(Y)$, which we compute by  means of the auxiliary filtration  $\Tilde F^\bu K_Q(Y)$. The first differential  has the effect of  setting $\lambda_5^- = - u$ on the $\Tilde E_1$ page,  and thus effectively inverting $\lambda_5^-$ and restricting to the complement of the corresponding hyperplane section. We can then solve five of  the equations  in~\eqref{eq:I-IIB}, obtaining
  \deq{
    \lambda_1^+ = u^{-1} \left( \lambda_1^- \lambda_5^+ + \lambda_3^+ \lambda_3^- \right).
  }
  Substituting this into the other equation produces
  \deq{
  u^{-1} \left( \lambda_1^- \lambda_5^+ + \lambda_3^+ \lambda_3^- \right) \lambda_3^- + \lambda_1^- \lambda_3^+ = 0.
  }
We observe that the  $(\lambda_3^-)^2\lambda_3^+$ term vanishes for  type reasons, so  that  we can simplify this equation as
  \deq{
    \lambda_1^- \left(  \lambda_3^+ + u^{-1} \lambda_5^+ \lambda_3^-  \right) = 0.
  }
  Defining $\tilde\lambda_3^+ = \lambda_3^+ + u^{-1} \lambda_5^+ \lambda_3^-$ and contracting the additional acyclic differentials we see can write the ideal $J$ generated by the remaining equations as 
  \[
  J = \<\lambda_1^- \tilde\lambda_3^+\>.
  \]
  We conclude that
  \[
  H^\bu(K_Q(Y)) \cong \CC[\theta_1^+] \otimes H^\bu(K(Y_Q))
  \]
  where $K(Y_Q) = \CC[\theta_3^+, \theta_1^-, \lambda_1^-, \Tilde{\lambda}_3^+] / J$ is the Koszul complex of the maximal ideal for the variety defined by the ideal $J$. 
\end{proof}

\subsection{Cohomology of $Y_Q$}
\label{sec:IIBkoszul}
By Lemma \ref{alemIIB} the cohomology of the associated graded of $(u)^{-1} A_Q$ is given by 
\beqn\label{eqn:IIBsuperfield}
\rC^\infty(\CC^5) [\theta_1^+] \otimes H^\bu (K(Y_Q))[u,u^{-1}] .
\eeqn
As above, since $Y_Q$ is cut out by homogenous equations its Koszul cohomology is bigraded with respect to $(\lambda, \theta)$ degree.

Let $\alpha = - \lambda_1^- \theta_3^+  - \Tilde{\lambda}_3^+\theta_1^-$. 
Representatives for the cohomology of this resolution are given as follows
\begin{itemize}
\item In degree $(0,0)$ the cohomology is one-dimensional with generator $1$. 
\item In degree $(1,1)$ the cohomology is five-dimensional with generator $\alpha$.  
\item In degree $(1,2)$ the cohomology is one-dimensional generated by $\theta_1^- \alpha$. 
\item In degree $(2,2)$ the cohomology is eleven-dimensional. 
One generator is given by $\Tilde{\lambda}_3^+ \theta_3^+ \alpha$. 
The other ten generators are $\alpha^2$.
\item In degree $(2,3)$ the cohomology is ten dimensional. 
Five generators are given by $\theta_1^- \alpha^2$.
Five generators are given by $\theta_3^+ \alpha^2$. 
\item In degree $(2,4)$ the cohomology is one-dimensional with generator $(\theta_1^-)^2 \alpha^2$.
\item In degree $(3,4)$ the cohomology is one-dimensional with generator $\theta_3^+ \alpha^3$. 
\end{itemize}

The remaining differential acting on the graded algebra \eqref{eqn:IIBsuperfield} coming from the spectral sequence abutting to $A_Q$ is of the form
\beqn\label{eqn:IIBdiff1}
\Tilde{\d} = - \lambda_1^- \theta_3^+ \frac{\partial}{\partial z} - \Tilde{\lambda}_3^+ \theta_1^- \frac{\partial}{\partial z} + u\theta_1^+ \frac{\partial}{\partial \zbar} .
\eeqn
From this, we can characterize the action of the differential on each of the classes above.
Let $f = f(z,\zbar, \theta_1^+)$ denotes an element of $\rC^\infty(\CC^5)[\theta_1^+]$.
Then
\beqn\label{eqn:IIBdiff2}
\Tilde{\d} (f) = (\partial_z f) \alpha + u (\partial_{\zbar} f) \theta_1 .
\eeqn
The value of the differential on the other generators is determined by the Leibniz rule. 

\section{Applications: twists of supergravity theories and supersymmetric field theories}

In this section, we recapitulate our results and reformulate them in the context of twists of supersymmetric field theories. 
We have reviewed above that the tautological filtered cdgsa $A$  provides a large free resolution of a particular multiplet over superspace, as in the  pure spinor superfield formalism. 
Only {\em minimal} twists of supersymmetric theories~\cite{CostelloHol, ESW} are relevant for us in this work. 
In what follows, we will indicate which multiplets are modelled by $A$ in each case, and give a component-field description of the twisted theory in the language of complex geometry. We will appeal to some ideas, both from the usual BV formalism and from supersymmetric physics, that were not reviewed in any detail in the remainder of the paper. 

\subsection{Supersymmetric gauge theory}

The results of \S \ref{sec:smooth} can be used to deduce the twists of minimally supersymmetric Yang--Mills theory in dimensions $2n = 4,6,10$. 
Characterizations of these twists are not new to this paper; they have appeared in \cite{Johansen, BaulieuCS, CostelloHol, CostelloYangian, CY4, ESW}, just to name a few references. 

The novel insight in our approach is that we describe the twisted theory in a rather universal way in terms of a {\em commutative} dg algebra. 
Given any Lie group $G$, we obtain the holomorphic twist of the perturbative supersymmetric Yang--Mills multiplet by tensoring with the Lie algebra $\fg$. (Note, however, that this multiplet may either be a BRST or BV multiplet, depending on dimension!)
Let us unwind this construction explicitly. 

Given any Lie algebra $\fg$ and a cdgsa $(B,\d_B)$ one can define a natural super dg Lie algebra structure on the tensor product $\fg \otimes B$. 
The differential is simply $\d = 1_\fg \otimes \d_B$ and the bracket is $[X \otimes a, Y \otimes b] = [X,Y] \otimes (ab)$. 
Furthermore, a map of cdgsa's $f \colon B \to B'$ induces a map of super Lie algebras $f \colon \fg \otimes B \to \fg \otimes B'$. 

Let $A$ be the tautological cdgsa associated to the minimal nilpotence variety in dimensions $4,6$ or $10$.
Following this construction, the tensor product $\fg \otimes A$ has the structure of super dg Lie algebra.
Furthermore, by Theorem \ref{thm:asmooth}, we obtain a quasi-isomorphism of super dg Lie algebras
\[
\Phi \colon (u)^{-1}\, \fg \otimes  \Dol(\CC^n) \xto{\simeq} (u)^{-1}\, \fg \otimes A_Q 
\]
where $A_Q$ is the twisted tautological cdgsa.

The super dg Lie algebra $(u)^{-1} \fg \otimes \Dol(\CC^n)$ can be thought of as a family of super dg Lie algebras over the punctured disk $u \in D^\times$. 
Following Remark~\ref{rmk:regrade} we can regrade this family turning $u$ into an element of $\ZZ \times \ZZ/2$ degree $(0,+)$.
Further, specializing $u=1$ results in the more familiar dg Lie algebra $\fg \otimes \Omega^{0,\bu}(\CC^n)$.
This dg Lie algebra describes the formal moduli space of holomorphic $G$-bundles on $\CC^n$ near the trivial holomorphic $G$-bundle. Our result exhibits this as a deformation of a dg Lie algebra describing the formal moduli space of field configurations in the corresponding supersymmetric gauge theory.

In dimensions four and six, supersymmetry closes off-shell, and the dg Lie algebra $\lie{g}\otimes A$ correspondingly describes the space of all field configurations modulo gauge invariance. In physics language, this is the BRST theory and does not yet contain information about dynamics. On the other hand, supersymmetry in the ten-dimensional theory only closes on the critical locus, so that $\lie{g} \otimes A$ can be given the structure of a full BV theory. For details, we refer the reader to~\cite{EHSW}.

\subsection{The six-dimensional tensor multiplet}

In \S \ref{sec:6d} we studied the twisted pure spinor superfield for six-dimensional $\cN=(2,0)$ supersymmetry. 
We refer back to that section for notation. 
Our analysis gives a description of the minimal twist of the abelian six-dimensional $\cN=(2,0)$ tensor multiplet, as studied using pure spinors in~\cite{CederwallM5}. The resulting multiplet was showed to admit the structure of a presymplectic BV theory by the authors in~\cite{SWtensor}, where the twist was also first computed in components. Comparing the lengths of the corresponding arguments will convince the reader of the efficiency of our approach here! In addition, the approach we use here reveals a hidden commutative algebra structure, analogous to that present in the gauge theory examples above.

Let $A_Q$ be the twisted tautological complex. 
Recall, by Theorem \ref{thm:a6d}, there is a quasi-isomorphism of cdgsa's 
\[
\Phi \colon (u)^{-1} B \xto{\simeq} (u)^{-1} A_Q
\]
where 
\[
  B = \Dol^{\leq 1}(\CC^3) \oplus \left( \Dol(\CC^3) \otimes \Pi R^\circ[-1]  \right) .
\]
Denote elements of $B$ by $(\alpha, \beta)$. 
Denote by $B_c \subset B$ the subalgebra of elements with compact support on $\CC^3$.
The algebra $B_c$ is equipped with a pairing of bidegree $(-5,+)$ defined by
\[
\omega (\alpha + \beta, \alpha' + \beta') = \int_{\CC^3} \alpha \partial \alpha' + \int_{\CC^3} \Omega \, \omega^\circ(\beta, \beta') 
\]
where $\Omega$ is the holomorphic volume form on $\CC^3$. 
In particular, the shift $B_c[2]$ is equipped with a pairing of bidegree $(-1, +)$. 

Following Remark \ref{rmk:regrade}, we can regrade $B[2]$, turning $u$ into an element of 
bidegree $(0,+)$.
Further, specializing $u=1$ results in the super cochain complex
\[
  \Omega^{\leq 1,\bu} (\CC^3)[2] \oplus \left( \Omega^{0,\bu}(\CC^3) \otimes \Pi R^\circ [1] \right).
\]
This is precisely the complex underlying the fields of the minimal twist of the {\em free} $\cN=(2,0)$ tensor multiplet, as found in \cite{SWtensor}.
The pairing $\omega$ endows this cochain complex with the structure of a {\em presymplectic BV} theory; see {\em loc.~cit.}\ for further details. 
This is one motivation for the seemingly arbitrary choice to consider $B[2]$ rather than some other regrading.

Because we have shifted the algebra $B $ by {\em two} rather than one, there is no obvious way to tensor with a Lie algebra to obtain an interacting BRST theory, like we did for the twist of supersymmetric Yang--Mills theory in the previous subsection.
In other words, simply tensoring $B$ with a Lie algebra clearly does not result in a non-trivial theory of a (twisted) non-abelian tensor multiplet. 

We nevertheless suspect that this algebra structure is involved in describing the gauge symmetries in the twist of the non-abelian tensor multiplet.
We leave a detailed investigation of this for future work, but remark that 
we expect this commutative structure to be combined in interesting fashion with objects in higher gauge theory studied in the theory of Lie $2$-algebras \cite{Saemann1, BaezHuerta}. 

\subsection{Eleven-dimensional supergravity}
\label{ssec:11dsugra}

In this case, it is known that $A[3]$, the shift by three of the tautological cdgsa resolves the untwisted eleven-dimensional multiplet~\cite{Ced-11d}.
As such, the results of \S\ref{sec:11d} provide a description of the minimal twist of eleven-dimensional supergravity. 
We emphasize that we only characterize the {\em free limit} of the twisted theory;
a full description of the interacting theory will be presented in forthcoming work with Surya Raghavendran. 

In \S \ref{ssec:koszul} we have characterized the spectral sequence abutting to the cohomology of $A_Q$ in terms of pure spinor variables. 
We unpack this using complex geometry, as we have done in the previous examples. 

Let $X$ be a Calabi--Yau $5$-fold and let $T_X, T^*_X$ denote the holomorphic tangent and cotangent bundles. 
Denote by 
\begin{align*}
  \Omega^{i,j} & = \Gamma\left(X,  \wedge^i T^*_X \otimes \wedge^j \Bar{T}_X^* \right), \\
\PV^{i,j} & = \Gamma\left(X,  \wedge^i T_X \otimes \wedge^j \Bar{T}_X^* \right)
\end{align*}
the spaces of type $(i,j)$ Dolbeault and polyvector fields, respectively. 
Also, let $\Omega^i$ be the space of (complex valued) smooth $i$-forms on $\RR$. 
Below, we will use the completed tensor product $\Hat{\otimes}$ with respect to the natural topology on spaces of sections; it has the property that 
\[
\Omega^{i,j} \, \Hat{\otimes} \, \Omega^k = \Gamma\left(\CC^5 \times \RR \, , \, \wedge^i T^*_{\CC^5} \otimes \wedge^j \Bar{T}_{\CC^5}^* \otimes T^*_\RR \right)
\]
and similarly for polyvector fields. 

By Lemma~\ref{alem11d}, the cohomology of $(u)^{-1} \Gr A_Q$ is given by 
\beqn
\rC^\infty(\CC^5 \times \RR) [\theta_0,\theta_1] \otimes H^\bu (K(Y_Q))[u,u^{-1}] .
\eeqn
We identify $\theta_0$ with the one-form $\d t$ and $\theta_1$ with the one-form $\d \zbar$. 
We denoted by $\Tilde{\d}$ the differential acting on the cohomology of this associated graded complex. 

Following Remark \ref{rmk:regrade}, we can regrade so as to turn $u$ into an element of bidegree $(0,+)$.
We then further specialize $u=1$ just as before. 
The differential $\Tilde{\d}$ from Equation \eqref{eqn:11diff1} is of degree $(1,+)$. 
This results in the $\ZZ \times \ZZ/2$ graded vector space
\[
  \rC^\infty(\CC^5 \times \RR) [\d t,\d \zbar] \otimes H^\bu(K(Y_Q))[3]. 
\]
We recall that the cohomology of $K(Y_Q)$ is graded by both $\lambda$- $\theta$-degree; we will correspondingly write $H^{i,j}$. Note, however, that the appropriate cdgsa bidegree for such an element is $(i,i+j\mod 2)$. We have recalled the generators of the Koszul homology in \S \ref{ssec:koszul}: 
\begin{itemize}
  \item $H^{0,0}(K(Y_Q))$ is one-dimensional, spanned by $1$.
We therefore have an $SL(5)$-equivariant identification 
\[
  \rC^\infty(\CC^5 \times \RR) [\d t,\d\zbar] \otimes H^{0,0} \cong \Omega^{0,\bu}(\CC^5) \, \Hat{\otimes} \, \Omega^\bu(\RR) .
\]
The overall bidegree is $(-3,+)$.
\item $H^{1,1}(K(Y_Q))$ is five-dimensional, spanned by $\alpha$.
We have an $SL(5)$ identification 
\[
  \rC^\infty(\CC^5 \times \RR) [\d t, \d\zbar] \otimes H^{1,1} \cong \Omega^{1,\bu}(\CC^5) \, \Hat{\otimes} \, \Omega^\bu(\RR) .
\]
The overall bidegree is $(-2, +)$.
\item $H^{1,2}(K(Y_Q))$ is five-dimensional, spanned by $\theta_2 \alpha$.
  As an $SL(5)$ representation,
\[
  \rC^\infty(\CC^5 \times \RR) [\d t, \d\zbar] \otimes H^{1,2} \cong \PV^{1,\bu}(\CC^5) \, \Hat{\otimes} \, \Omega^\bu(\RR) .
\]
The overall bidegree is $(-2,-)$.
\item $H^{2,3}(K(Y_Q))$ is one-dimensional, spanned by $\theta_2 \alpha^2$.
We have an $SL(5)$ identification 
\[
  \rC^\infty(\CC^5 \times \RR) [\d t,\d\zbar] \otimes H^{2,3} \cong \PV^{0,\bu}(\CC^5) \, \Hat{\otimes} \, \Omega^\bu(\RR) .
\]
The overall bidegree is $(-1,-)$.
\end{itemize}

From the characterization of the differential $\Tilde{\d}$ in Equation \eqref{eqn:11diff2} we obtain the following cochain complex
\beqn\label{eqn:11dsugradiagram}
\begin{tikzcd}[row sep = 4 pt]
  \ul{-3} & \ul{-2} & \ul{-1}  \\
& \Pi \PV^{1,\bu} \, \Hat{\otimes} \, \Omega^\bu \ar[r,"\partial_\Omega"] & \Pi \PV^{0,\bu} \, \Hat{\otimes} \, \Omega^\bu  \\
\Omega^{0,\bu} \, \Hat{\otimes} \, \Omega^\bu  \ar[r, "\partial"] & \Omega^{1,\bu} \, \Hat{\otimes} \, \Omega^\bu   .
\end{tikzcd}
\eeqn
where $\partial_\Omega$ denotes the holomorphic divergence operator with respect to the flat Calabi--Yau structure on $\CC^5$ and $\partial$ is the holomorphic de Rham operator. 
The $\dbar$ and $\d_t$ operators are left implicit, as always.

The theory admits a nondegenerate pairing of degree $(-2,-)$. As such, it does not straightforwardly admit the structure of a BV theory. However, we note that it is a $\Z/2\Z$-graded BV theory upon totalizing the cdgsa bidegree. 

\subsection{Type IIB supergravity}
Motivated by the topological string, Costello and Li have laid out a series of conjectures for twists for Type IIB and Type IIA supergravity theories \cite{CLsugra} in terms of Kodaria--Spencer theory \cite{BCOV}. 
The results in \S \ref{sec:IIB} confirm their conjecture for the free limit of the minimal twist of Type IIB supergravity on $\CC^5$. 
In the terminology of \cite{CLsugra} we find the free limit of the ``minimal'' form of Kodaira--Spencer theory. 

To match up with \cite{CLsugra} it is necessary to consider the overall cohomological shift of the tautological complex by four $A[4]$. 
We remark on some consequences of this choice at the end of the subsection.
By Lemma~\ref{alem11d}, the cohomology of $(u)^{-1} \Gr A_Q[4]$ is given by 
\beqn
\rC^\infty(\CC^5 \times \RR) [\theta_1^+] \otimes H^\bu (K(Y_Q))[u,u^{-1}] [4] .
\eeqn
We identify $\theta_1^+$ with the one-form frame $\d \zbar$.
We denoted by $\Tilde{\d}$ the differential acting on the cohomology of this associated graded complex. 

Following Remark \ref{rmk:regrade}, we can regrade so as to turn $u$ into an element of bidegree $(0,+)$.
We then further specialize $u=1$ just as before. 
The differential $\Tilde{\d}$ from Equation \eqref{eqn:11diff1} is of degree $(1,+)$. 
This results in the $\ZZ \times \ZZ/2$ graded vector space
\beqn\label{eqn:IIBregrade}
  \rC^\infty(\CC^5 \times \RR) [\d \zbar] \otimes H^\bu(K(Y_Q))[4]. 
\eeqn
We recall that the cohomology of $K^\bu(Y_Q)$ is graded by both $\lambda$- $\theta$-degree; we will correspondingly write $H^{i,j}$. Note, however, that the appropriate cdgsa bidegree for such an element is $(i,i+j\mod 2)$. We have listed the generators of the Koszul homology in \S \ref{sec:IIBkoszul}. 

Proceeding to make $SL(5)$ identifications as we did in the previous subsection, we see that \eqref{eqn:IIBregrade} equipped with the differential $\Tilde{\d}$ takes the following form:
\[
  \begin{tikzcd}[row sep = 4 pt]
\ul{-4} & \ul{-3} & \ul{-2} & \ul{-1}  \\
& & \PV^{0,\bu} \\
& & \Pi \PV^{1,\bu} \ar[r, "\partial_\Omega"] & \Pi \PV^{0,\bu} \\
\Omega^{0,\bu} \ar[r, "\partial"] & \Omega^{1,\bu} \ar[r, "\partial"] & \Omega^{2,\bu} &  & \\
& \Pi \Omega^{0,\bu} \ar[r, "\partial"] & \Pi \Omega^{1,\bu} & & \\
& & \Omega^{0,\bu}
\end{tikzcd}
\]
Here, we have used the characterization of the differential $\Tilde{\d}$ in Equation \eqref{eqn:IIBdiff2}.

This complex admits a (degenerate) pairing of degree $(-1, +)$; this is the reason for the seemingly arbitrary overall shift by four. 
Away from the middle line, this is the obvious pairing granted by the Calabi--Yau structure. 
On the middle line the pairing is defined by 
\[
\int_{\CC^5} \alpha \partial \alpha'
\]
for $\alpha, \alpha' \in \Omega^{2,\bu}$. 
This results in the structure of a presymplectic BV theory \cite{SWtensor}. 

Notice that the last three lines of the above complex are of the form $\Omega^{\leq i, \bu}$; this is the complex resolving the sheaf of holomorphic closed $i$-forms on $\CC^5$.
The $\partial$-operator defines a map
\[
\partial \colon \Omega^{\leq i, \bu} \to \Omega^{>i , \bu} \cong \PV^{\leq 5-i, \bu}
\]
where the isomorphism makes use of the Calabi--Yau structure. 
If one replaces $\Omega^{\leq i, \bu}$ in the complex with $\PV^{\leq 5-i, \bu}$ one finds precisely the description given in \cite{CLsugra}.\footnote{We have a $\ZZ/2 \times \ZZ$ graded complex which is broken to the totalized $\ZZ/2$ grading in the presence of the BCOV interaction.}
This complex is no longer is equipped with a pairing, but is equipped with a degree $(-1,+)$ Poisson bivector.
This bivector produces the so-called BCOV propagator which Costello and Li use to study the string field theory quantization of the topological B-model \cite{CLbcov1}. 

\subsection{Type IIA supergravity}
Physically, Type IIA supergavity is related to eleven-dimensional supergavity by circle compactification. 
Using this intuition and the discussion in \S \ref{ssec:11dsugra}, we provide a conjectural description of the free limit of the {\em minimal} twist of Type IIA supergravity on $\CC^5$. 

We compactify along a circle in the topological direction. 
This amounts to replacing the de Rham complex $\Omega^\bu$ in \eqref{eqn:11dsugradiagram} by the de Rham cohomology of the circle $\CC \oplus \Pi \CC$, where the odd generator is the volume form on $S^1$. 
The result is
\[
  \begin{tikzcd}[row sep = 4 pt]
\ul{-3} & \ul{-2} & \ul{-1} &  \\
& \Pi \PV^{1,\bu}  \ar[r,"\partial"] & \Pi \PV^{0,\bu}   \\
&  \PV^{1,\bu}  \ar[r,"\partial"] & \PV^{0,\bu}   \\
\Omega^{0,\bu}  \ar[r, "\partial"] & \Omega^{1,\bu}  \\
\Pi \Omega^{0,\bu}  \ar[r, "\partial"] & \Pi \Omega^{1,\bu}  .
\end{tikzcd}
\]

\subsection{Type I supergravity}

Finally, we provide a speculative description of the minimal twist of the minimal twist of Type I supergravity. 
We obtain this simply by comparing the field content in the Type IIA and Type IIB models above. 
We observe that these complexes share the subcomplex:
\[
  \begin{tikzcd}[row sep = 4 pt]
\ul{-3} & \ul{-2} & \ul{-1} &  \\
& \Pi \PV^{1,\bu}  \ar[r,"\partial_\Omega"] & \Pi \PV^{0,\bu}   \\
\Pi \Omega^{0,\bu}  \ar[r, "\partial"] & \Pi \Omega^{1,\bu}  .
\end{tikzcd}
\]
The bottom line maps via $\partial$ to the complex 
\[
\Pi \PV^{3, \bu} \to \Pi \PV^{2, \bu} \to \cdots \to \Pi \PV^{0,\bu}
\]
Upon doing this, one finds a match with Costello and Li's conjectural description of Type I supergravity which they arrived at from studying the unoriented version of the $B$-model~\cite{CLtypeI}.
\printbibliography

\end{document}